\documentclass[acmsmall,authorversion]{acmart}

\usepackage[utf8]{inputenc}
\usepackage{amsmath, amsthm}
\usepackage{amsfonts}
\usepackage{amssymb}
\usepackage[only,coloneqq]{mathtools}
\usepackage{float}
\usepackage{mathrsfs}
\usepackage{tikz}

\usepackage{microtype}

\usepackage[shortlabels,inline]{enumitem}
\setlist[description]{labelsep=1em,font=\normalfont\bfseries}

\usepackage{todonotes}

\usepackage[]{algpseudocode}

\floatstyle{ruled}
\newfloat{algo}{tp}{lop}
\floatname{algo}{Algorithm}
\def\eproof{{\mbox{}\hfill\qed}\medskip}

\usepackage{color}
\definecolor{red}{rgb}{.7,0,0}
\definecolor{blue}{rgb}{0,0,1}

\usepackage[labelfont=bf,labelsep=space]{caption}

\newtheorem{theorem}{Theorem}[section]
\newtheorem{proposition}[theorem]{Proposition}
\newtheorem{lemma}[theorem]{Lemma}
\newtheorem{corollary}[theorem]{Corollary}

\theoremstyle{definition}
\newtheorem{definition}[theorem]{Definition}

\newtheorem{remark}[theorem]{Remark}

\newcommand{\e}{\varepsilon}

\newcommand{\bfa}{{\boldsymbol a}}

\newcommand{\SigmaC}{\Sigma^\C}
\newcommand{\Hc}{\mathcal{H}}

\def\mun{\mu_{\mathrm{norm}}}
\def\kappan{\kappa_{\mathrm{norm}}}

\def\mup{\mu_{\mathrm{proj}}}

\def\kappasa{\kappa_{*}}
\def\kappaaff{\kappasa^\mathrm{aff}}

\def\Approx{\mathsf{Approx}}

\def\hm{^{\mathsf{h}}}
\def\Sigmaaff{\Sigma_*^\mathrm{aff}}
\def\Diff{\mathrm D}

\renewcommand{\tilde}{\widetilde}

\def\N{\mathbb{N}}
\def\Z{\mathbb{Z}}
\def\R{\mathbb{R}}

\def\C{\mathbb{C}}

\def\IS{\mathbb{S}}

\newcommand{\algoritmo}{\begin{minipage}{0.87\hsize}\linea}
\newcommand{\falgoritmo}{\linea\end{minipage}\bigskip}
\newcommand{\codigo}{\begin{minipage}{0.87\hsize}}
\newcommand{\fcodigo}{\end{minipage}\bigskip}
\newcommand{\linea}{\vspace*{-5pt}\hrule\vspace*{5pt}}

\def\mop{\operatorname}

\def\rank{\mop{rank}}

\def\Id{\mop{id}}

\def\cost{\mop{cost}}

\def\Prob{\mathop{\mop{Prob}}}

\def\bd{{\boldsymbol d}}
\def\Oh{{\mathcal O}}

\def\appleq{\hbox{\lower3.5pt\hbox{$\;\:\stackrel{\textstyle<}{\sim}\;\:$}}}

\def\ud{\mathrm{d}}

\def\eqdef{\coloneqq}
\def\epsilon{\varepsilon}

\newcommand{\bfd}{{\boldsymbol{d}}}

\def\mcO{{\mathcal O}}
\def\mcH{{\mathcal H}}

\def\mcN{{\mathcal N}}
\def\mcG{{\mathcal G}}
\def\mcU{{\mathcal U}}
\def\mcM{{\mathcal M}}

\def\mcX{{\mathcal X}}

\def\mcT{{\mathcal T}}

\def\mcP{{\mathcal P}}

\def\scO{{\mathscr O}}
\def\scU{{\mathscr U}}

\def\Hd{\mcH_{{\boldsymbol{d}}}}
\def\HdC{\mcH^{\C}_{{\boldsymbol{d}}}}
\def\Pd{\mcP_{\!{\boldsymbol{d}}}}

\def\st{\ \middle|\ }
\def\geq{\geqslant}
\def\leq{\leqslant}
\def\le{\leqslant}
\def\ge{\geqslant}

\def\P{\mathbb P}

\title[Computing the Homology of Basic Semialgebraic Sets]{Computing the Homology of Basic Semialgebraic Sets\\
in Weak Exponential Time}

\author{Peter B\"urgisser}
\affiliation{%
  \institution{Technische Universit\"at Berlin}
  \department{Institut f\"ur Mathematik}
  \country{Germany}}
\email{pbuerg@math.tu-berlin.de}

\author{Felipe Cucker}
\affiliation{%
  \institution{City University of Hong Kong}
  \department{Department of Mathematics}
  \state{Hong Kong}}
\email{macucker@cityu.edu.hk}

\author{Pierre Lairez}
\affiliation{%
  \institution{Inria}
  \country{France}}
\email{pierre.lairez@inria.fr}



\acmJournal{JACM}
\acmVolume{66}
\acmNumber{1}
\acmArticle{5}
\acmYear{2018}
\acmDOI{10.1145/3275242}
\setcopyright{none}

\begin{document}

\begin{abstract} We describe and analyze an 
algorithm for computing the homology (Betti numbers and torsion 
coefficients) of basic semialgebraic sets which works in weak 
exponential time. That is, out of a set of exponentially small 
measure in the space of data, the cost of the algorithm 
is exponential in the size of the data. All algorithms previously 
proposed for this problem have a complexity which is doubly 
exponential (and this is so for almost all data). 
\end{abstract}

\maketitle

\section{Introduction}

Semialgebraic sets (that is, subsets of Euclidean spaces defined by
polynomial equations and inequalities with real coefficients) 
come in a wide variety of
shapes and this raises the problem of describing a given specimen,
from the most primitive features, such as emptiness, dimension, or
number of connected components, to finer ones, such as roadmaps,
Euler-Poincar\'e characteristic, Betti numbers, or torsion
coefficients. 

The {\em Cylindrical Algebraic Decomposition} (CAD) introduced by
Collins~\cite{Collins} and W\"uthrich~\cite{Wut76} in the 1970's 
provided algorithms to compute these features that worked 
within time $(sD)^{2^{\Oh(n)}}$ where~$s$ is the number of 
defining equations, $D$ a bound on their degree
and~$n$ the dimension of the ambient space. Subsequently, 
a substantial effort was devoted to design algorithms for these 
problems with \emph{single exponential} algebraic complexity 
bounds~\cite[][and references therein]{Basu:08b}, that is, 
bounds of the form~$(sD)^{n^{\Oh(1)}}$. 
Such algorithms have been found for deciding
emptiness~\cite{GrigorevVorobjov_1988,Ren92a,BaPoRo96}, 
for counting connected
components~\cite{BaPoRo99,canny:93,canny-grig-voro:92,GriVo92,heintz-roy-solerno-94},
computing the dimension~\cite{Koi98,BasuPollackRoy_2006a}, the Euler-Poincar\'e
characteristic~\cite{Basu_1996}, the first few Betti numbers~\cite{Basu_2006},
the top few Betti numbers~\cite{Basu:08} and roadmaps (embedded curves with certain
topological properties)~\cite{Canny_1991,BasuRoySafeyElDinEtAl_2014,SafeyElDinSchost_2017}.

As of today, however, no single 
exponential algorithm
is known for the computation of the whole sequence of the 
homology groups (Betti numbers and torsion coefficients). 
For complex smooth projective varieties,
Scheiblechner~\cite{Scheib:12} has been able to provide  
an algorithm computing the Betti numbers (but not the torsion 
coefficients) in single exponential time relying on the algebraic 
De~Rham cohomology. The same author provided a lower
bound for this problem (assuming integer coefficients) 
in~\cite{Scheib:07}, where the problem is shown to 
be \textsc{PSpace}-hard.

Another line of research, that has developed independently of the
results just mentioned, focuses on the complexity and the 
geometry of
numerical algorithms~\cite[and references therein]{Condition}. The
characteristic feature of these algorithms is the use of
approximations and successive refinements. For most problems, 
a set of
{\em numerically ill-posed} data can be identified, for which
arbitrarily small perturbations may produce qualitative errors in the
result of the computation. Iterative numerical algorithms may run
forever on ill-posed data, and may take increasingly long running 
time as data become close to ill-posed. The running time is 
therefore not
bounded by a function on the input size only and the usual 
worst-case analysis is irrelevant. An alternate form of analysis, 
championed by Smale~\cite{Smale97} and going back 
to~\cite{hest-stiefel:52}, bounds the running time of an algorithm 
in terms of the size of the input and a {\em condition number}, usually 
related to, or bounded by, the inverse to the distance of the data 
at hand to the set of ill-posed data.

Then, the most common way to gauge the complexity 
of a numerical algorithm is to endow the space of data with 
a probability measure, usually the standard Gaussian, and
to analyze the algorithm's cost in probabilistic terms.
More often than not, this analysis results in a bound on the 
expectation of the cost, that is, in an {\em average-case analysis}. 
But recently, Amelunxen and
Lotz~\cite{wacco} introduced a new way of measuring complexity. 
They noticed that a number of algorithms that are known to be 
efficient in practice have
nonetheless a large, or even infinite, average-case complexity. 
One of the reasons they identified for this discrepancy 
is the exponentially fast vanishing measure of an exceptional
set of inputs on which the algorithm runs in superpolynomial 
time when the dimension grows. A prototype of this phenomenon 
is the behavior of the power method to compute a dominant 
eigenpair of a symmetric matrix. This algorithm is
considered efficient in practice, yet it has been shown that 
the expectation of the number of iterations performed by the 
power method, for matrices drawn from
the orthogonal ensemble, is infinite~\cite{Kostlan88}. 
Amelunxen and Lotz show that, conditioned to exclude a set of
exponentially small measure, this expectation is $\Oh(n^2)$
for~$n\times n$ matrices. The moral of the story is that the power
method is efficient in practice because it is so in theory if we
disregard a vanishingly small set of outliers. This conditional
expectation, in the terminology of~\cite{wacco}, shows a {\em weak
average polynomial cost} for the power method. More generally, we will
talk about a complexity bound being {\em weak} when this bound holds
out of a set of exponentially small measure.

Several problems related to semialgebraic sets have been studied from
the numerical point of view we just described, such as deciding
emptiness~\cite{CS98}, counting real solutions of zero-dimensional
systems~\cite{CKMW1}, or computing the homology groups of real
projective sets~\cite{CuckerKrickShub_2018}. Our main result follows
this stream.

\paragraph{Main result.}

A \emph{basic semialgebraic set} is a subset of a Euclidean space $\R^n$ 
given by a system of equalities and inequalities of the 
form 
\begin{equation}
  \label{eq:sa}
   f_1(x)= \dotsb = f_q(x) = 0 \text{ and } 
  g_1(x) \succ 0, \dotsc, g_s(x)\succ 0
\end{equation}
where~$F = (f_1,\dotsc,f_q)$ and~$G = (g_1,\dotsc,g_s)$ are 
tuples of polynomials with real coefficients and the expression 
$g(x)\succ 0$ stands for either $g(x)\geq0$ or $g(x)>0$ 
(we use this notation to emphasize the fact, that will 
become clear in~\S\ref{se:relax}, 
that our main result does not 
depend on whether the inequalities in~\eqref{eq:sa} are strict).  
Let $W(F,G)$ denote the solution set of the semialgebraic 
system~\eqref{eq:sa}. 

For a vector~$\bd=(d_1,\dotsc,d_{q+s})$ of~$q+s$ positive integers, 
we denote by~$\Pd$ (or~$\Pd[q\,;\,s]$ to emphasize the number 
of components)
the linear space of the~$(q+s)$-tuples of real polynomials in~$n$
variables of degree at most~$d_1,\dotsc,d_{q+s}$, respectively. 
Let~$D$ denote
the maximum of the~$d_i$. We will assume that $D\ge2$ because 
a set defined by degree~1 polynomials is convex and its homology is
trivial. Let~$N$ denote the dimension of~$\Pd$, that is, 
\begin{equation}\label{eq:size}
   N=\sum_{i=1}^{q+s} \binom{n+d_i}{n}. 
\end{equation}
This is the \emph{size} of the semialgebraic system~\eqref{eq:sa}, 
as it is 
the number of real coefficients necessary to determine it. We endow
$\Pd$ with the Weyl inner product and its induced norm, 
see~\S\ref{sec:cond-semi-sets}. We further endow $\Pd$ with the 
standard Gaussian measure given by the density
${(2\pi)^{-\frac{N}{2}}}\exp\big({-\frac{\|(F,G)\|^2}{2}}\big)$ (we note, 
however, that we could equivalently work with the uniform 
distribution on the unit sphere in~$\Pd$). Finally, 
we distinguish a subset $\Sigmaaff$ of $\Pd[q\,;\,s]$ of 
ill-posed data (see~\S\ref{sec:cond-numb-affine} 
for a precise definition). 
Pairs $(F,G)$ on this set are those for which the Zariski 
closure in $\R^n$ of one of the algebraic sets defining the 
boundary of $W(F,G)$ is not smooth. We will see that 
$\Sigmaaff$ is a hypersurface in $\Pd[q\,;\,s]$ and hence has 
measure zero.

The complexity model we consider is the usual Blum--Shub--Smale model
 \cite{BlumShubSmale_1989}
extended (as often) with the ability to compute square roots.
What we call ``numerical algorithm'' is a machine in this model.

\begin{theorem}\label{thm:Main} 
There is a numerical algorithm~{\scshape Homology} that, given a system $(F, G)\in\Pd$ 
with~$q \leq n$ equalities and~$s$ inequalities, computes the 
homology groups of\/~$W(F, G)$.
Moreover, the number of arithmetic operations in~$\R$ performed 
by \textsc{Homology} on
input~$(F,G)$, denoted $\mop{cost}(F,G)$, satisfies
\begin{enumerate}[(i)]
\item \label{item:9} 
$\mop{cost}(F,G) = \big((s+ n) D \delta^{-1}\big)^{\Oh(n^2)}$ 
where $\delta$ is the distance of~$\frac{1}{\|(F,G)\|} (F,G)$ to 
$\Sigmaaff$.
\end{enumerate}
Furthermore, if~$(F, G)$ is drawn from the Gaussian measure 
on $\Pd$, then
\begin{enumerate}[(i),resume]
\item \label{item:7} $\cost(F, G)\leq \big((s+n)D\big)^{\Oh(n^3)}$ 
with probability at least $1-\big((s+n)D\big)^{-n}$
\item \label{item:8} $\cost(F, G)\leq 2^{\Oh(N^2)}$ with probability 
at least $1-2^{-N}$.
\end{enumerate}
The algorithm is numerically stable.
\end{theorem}

\emph{Nota bene.} The notation $\Oh$ will always be 
understood with respect to~$N$. For example, the 
bound $\mop{cost}(F,G) \leq \big((s+ n) D \delta^{-1}
\big)^{\Oh(n^3)}$ rewords as 
$\mop{cost}(F,G) \leq \big((s+ n) D \delta^{-1}
\big)^{C n^3}$ for some~$C > 0$ as soon as $N$ is large enough, 
even if some of the parameters $s$, $n$ or~$D$ are fixed.

Point~\ref{item:7} above does not imply, strictly speaking, weak
exponential time because for given $n$, $q$, $s$ and $\bd$, the
measure of the exceptional set is bounded by $\big((s+n)D\big)^{-n}$
and this may not be exponentially small in the input size $N$ (for
instance, when $n$ is fixed and $D$ and $s$ grow). But
Point~\ref{item:8} shows exactly what we can call \emph{weak 
exponential complexity}: out of an exponentially small set in the 
space of data the cost of the
latter is bounded by a single exponential function on the input size.

It is difficult to compare our algorithm with previous ones: because
of its numeric nature, it only deals with the generic case, at
positive distance from ill-posed problems, and its worst-case
complexity is unbounded. Nevertheless, it compares favorably 
with the doubly exponential worst-case bound obtained from 
the CAD. The latter is reached on generic inputs, whereas we 
show a single exponential
worst-case complexity outside a vanishingly small subset.
Another difference with previous works is the fact that 
our results are valid only for polynomials with real coefficients 
as our proofs use some analytical techniques. 
This is in contrast with, for instance, the work by 
Basu~\cite{Basu_2006,Basu:08} where the results are valid 
for semialgebraic sets defined over arbitrary real closed fields.

In this work, we approach the topology of a set by approximating it by a union
of Euclidean balls in the ambient space, as initiated in the field of
topological data analysis \cite[e.g.]{EdelsbrunnerMucke_1994}. Following ideas
in~\cite{NiyogiSmaleWeinberger_2008,ChazalCohen-SteinerLieutier_2009}, for the
coverings, we choose a union of balls of sufficiently small radius to which we
can apply the Nerve Theorem. In constrast, previous work by Basu et al.
\cite{BaPoRo:08,BasuPollackRoy_2005,Basu_2006}, with a more algebraic flavor,
approaches the topology from inside by computing covering by contractible
subsets of the original set.

Lastly, we note that all the ingredients in algorithm {\scshape
Homology} easily parallelize. Doing so, we obtain a parallel algorithm
working in \emph{weak parallel polynomial complexity}: out
of an exponentially small set in the space of data 
the parallel cost of the algorithm is bounded
by a polynomial function on the input size. The
\textsc{PSpace}-hardness result by 
Scheiblechner~\cite{Scheib:07} mentioned
above (together with the classical equivalence between 
space and parallel time~\cite{Borodin77}) suggests that 
further complexity improvements are limited as they are unlikely 
to be below parallel polynomial time.

\paragraph{Overview.}
This article follows some algorithmic ideas (grid methods, theory of
point estimates) introduced by Cucker and Smale~\cite{CS98} and
extended by
Cucker et~al.~\cite{CKMW1,CuckerKrickShub_2018}. 
In particular, an algorithm for computing the homology of a real
{\em algebraic} subset of $\IS^n$ (defined with only equalities, as 
opposed to \emph{semialgebraic} sets) has been studied
in~\cite{CuckerKrickShub_2018}. The spirit and the statement of our
main result is very close to this previous work but the methods are
substantially renewed.
There is a significant overlap with \cite{CuckerKrickShub_2018}
where we felt that the theory could be simplified (\S \ref{sec:proof-tau-gamma}
and \S \ref{sec:number-kappa}), but the specificity of the semialgebraic case
called for the application of different tools, such as the reach (\S
\ref{sec:appr-sets-with}), continuous Newton method (\S
\ref{sec:cont-alpha-theory}), or the relaxation of semialgebraic inequalities (\S \ref{sec:exclusion-inclusion}).

Besides, the numerical stability of the
algorithm in Theorem~\ref{thm:Main} will not be discussed here. The
precise meaning of this stability and its proof are a straightforward
variation of the arguments detailed
in~\cite[\S7]{CuckerKrickShub_2018} which in turn are based on those
in~\cite{CS98,CKMW1}.

Our method relies on several quantities reflecting corresponding aspects of the
conditioning of a semialgebraic system. The first one is the \emph{reach}. This
is a measure of curvature for sets without the structure of a manifold.
The second one measures how much the solution set of a 
semialgebraic system is affected by small relaxations of the 
equalities and inequalities of the system. 
The third one is the condition number $\kappasa$
which reflects the distance of a semialgebraic system to the 
closest ill-posed
system. The facts that $\kappasa$ bounds the other two measures 
and that we can
compute it efficiently are cornerstones of our algorithm. In a number of
respects, this condition number is a natural extension of the first 
instances of this notion, for systems of linear equations, introduced by
Turing~\cite{Turing48} and von~Neumann and 
Goldstine~\cite{vNGo47}. 

Sections~2, 3 and~4 study all these notions. They decrease in the 
generality of the context (closed sets, analytic sets, and semialgebraic 
sets, respectively) but increase on the computational use of the results.

In a few words, to compute the homology group of an arbitrary 
basic semialgebraic set~$W$, we first reduce to the case of a 
closed semialgebraic subset~$S$ of a
sphere~$\IS^n$. Then we gather a finite set~$\mcX$ of points
in~$\IS^n$ that is sufficiently dense and retain only the points that
are close enough to~$S$. A point is close enough to~$S$ if it
satisfies the defining equations and inequalities of~$S$ up to
some~$\epsilon$.  Extending a theorem of
Niyogi, Smale and Weinberger~\cite{NiSmWe08}, we argue that this finite set
of close enough points is sufficient to compute the homology
of~$S$. The condition number~$\kappasa$ acts as a master
parameter: it controls the meaning of ``sufficiently dense'' and
``close enough'' and, beyond that, the total complexity of the
algorithm and the required precision to run it. 

Besides the main result, this work features several notable
contributions.  First, an extension to sets with positive reach of the
Niyogi-Smale-Weinberger theorem about the computation of the 
homotopy type of a set via an approximation with a finite set
(Theorem~\ref{thm:SNW}). Second, a continuous analogue 
of Shub and
Smale's $\alpha$-Theorem in which Newton's iteration is replaced with
Newton's flow (Theorem~\ref{thm:cont-alpha}). Third, an inequality
relating the reach and the $\gamma$-number at a point of a real
analytic set (Theorem~\ref{thm:tau-gamma-local}). This
strenghtens and simplifies a result of
Cucker, Krick and Shub~\cite{CuckerKrickShub_2018}. 
Four, a theory of the conditioning of a semialgebraic system relating the
distance to the closest ill-posed problem to the sensitivity of the
solution set to small relaxations of the equalities and inequalities 
of the system (Theorem~\ref{thm:approx}). 
This is reminiscent of the Eckhart--Young theorem
for linear systems.

\section{Approximation of sets with positive reach}
\label{sec:appr-sets-with}

The \emph{reach} of a closed subset of a Euclidean space~$E$ is a
notion introduced by Federer~\cite{Federer_1959} to quantify the
curvature of objects without the structure of a manifold. We establish
a few useful properties of the reach and we use this notion to extend
a theorem of
Niyogi, Smale and Weinberger~\cite{NiSmWe08} that gives a criterion 
to compute the topology of a compact subset of an Euclidean 
space by means of a finite covering of balls with the same radius 
(Theorem~\ref{thm:SNW}). It will play a 
fundamental role in our arguments.

\subsection{Measures of curvature}		
\label{sec:cond-numb-clos}					   
For a nonempty subset $W\subseteq E$ and $x\in E$, 
we denote by $d_W(x) :=\inf_{p\in W} \|x-p\|$
the distance of~$x$ to~$W$. 
We note that the function $d_W \colon E\to\R$ is
$1$-Lipschitz continuous, that is, $|d_W(x) -d_W(y)| \le \|x-y\|$ 
for all $x,y\in E$.

\begin{definition}\label{def:reach}
Let $W\subseteq E$ be a nonempty closed subset. 
The {\em medial axis} of $W$ is defined as the closure of the set 
\begin{equation*}
    \Delta_W := \left\{ x\in E \st \exists p,q\in W, p\neq q
    \text{ and }\|x-p\|=\|x-q\|=d_W(x) \right\}.
\end{equation*}
The \emph{reach (or local feature size) of~$W$ at a point~$p\in W$} 
is defined as $\tau(W,p) \eqdef d_{\Delta_W}(p)$.
The \emph{(global) reach of~$W$} is defined as 
$\tau(W) \eqdef \inf_{p\in W} \tau(W,p)$.
We also set $\tau(\varnothing) := +\infty$.
\end{definition}

Note that $\tau(W)$ is also given by~$\inf_{x\in \Delta_W} d_W(x)$.
We can also characterize $\tau(W)$ as the supremum of all~$\epsilon$
such that for every~$x\in E$ with~$d_W(x) <\epsilon$, there exists a
unique point~$p\in W$ with~$\|x-p\| = d_W(x)$. We shall denote this
unique point by $\pi_W(x)$. This gives a map $\pi_W \colon T(W) \to
W$, where $T(W) := \big\{ x\in E \mid d_W(x) < \tau(W) \}$ denotes the
open neighborhood of $W$ with radius $\tau(W)$.

When $W$ is a smooth submanifold of~$E$, the reach of~$W$ can be
characterized in terms of the normal bundle of~$W$ as follows. Let
$N_\epsilon(W) := \{ (x,v) \in W\times E \mid v \perp T_x W , \|v\| <
\epsilon\}$ denote the open normal bundle of~$W$ with radius
$\epsilon$. The reach~$\tau(W)$ is the supremum of all~$\epsilon$ such
that the map $N_\epsilon(W) \to T(W),\, (x,v) \mapsto x + v$, is
injective~\cite{NiSmWe08}.

\begin{proposition}\label{pro:piW}
If~$\tau(W) > 0$, then $\pi_W \colon T(W) \to W$ is continuous
and the map
\begin{equation*}
    T(W)\times[0,1] \to T(W),\, (x,t) \longmapsto t\pi_W(x) + (1-t)x 
\end{equation*}
is a deformation retract of~$T(W)$ onto~$W$.
\end{proposition}

\begin{proof}
Concerning the continuity of~$\pi_W$,
let $(x_k)_{k\geq0}$ be a sequence in $T(W)$ converging 
to some $x\in T(W)$. We have 
$$  
\|\pi_W(x_k)-x\| \leq \|\pi_W(x_k) -x_k \| + \| x_k - x\| 
                 = d_W(x_k) + \| x_k - x\|  
                 \leq d_W(x) + 2 \| x_k - x\|,
$$
where we used the Lipschitz continuity of $d_W$ for the 
last inequality. 
Hence the sequence $\pi_W(x_k)$ is bounded. Let $y\in W$ be a 
limit point of $\pi_W(x_k)$. The above inequality implies that 
$\|y-x\| \leq d_W(x)$, hence
$y=\pi_W(x)$. Thus $\pi_W(x)$ is the only limit point of the 
sequence $\pi_W(x_k)$ and therefore, 
$\lim_{k\to\infty} \pi_W(x_k) = \pi_W(x)$.

The second claim is obvious.
\end{proof}

We will use the following well-known fact. 

\begin{lemma}\label{le:ortho} 
Assume there is an open neighborhood $U$ of $\pi_W(x)$, $x\in E$, 
such that $W\cap U$ is a smooth submanifold of $E$. 
Then $\pi_W(x)-x$ is normal to the tangent space of~$W$ 
at $\pi_W(x)$.\eproof
\end{lemma}

The main result of this section is a lower bound on 
the reach of an intersection $W\cap V$ in terms of the reach 
of $W$ and the reach of the intersection of~$W$ with the 
boundary~$\partial V$ of $V$. 

\begin{theorem}\label{thm:tau-boundary}
For closed subsets $V, W$ of~$E$ we have 
$\tau(W \cap V) \geq \min ( \tau(W),\tau(W\cap \partial V))$. 
\end{theorem}

For the proof, we introduce an auxiliary notion.
Let $W \subseteq E$ be a closed subset and $p\in W$.
Moreover, consider $u\in E$ with $\|u\|=1$. 
It is easy to see that $\{ t\ge 0 \mid  d_W(p+ t u ) = t\}$ 
is an interval containing~$0$. 
We are interested in those directions~$u$, where 
this interval has positive length and define the 
\emph{reach $\tau(W,p,u)$ of $W$ at $p$ along direction~$u$} 
as the length of this interval, that is,   
\begin{equation*}
 \tau(W,p,u) \eqdef \sup\left\{ t \geq 0 \st d_W(p+ tu) = t \right\} .
\end{equation*}
We note that $\pi_W(p+tu) = p$ for any~$0 \le t<\tau(W,p,u)$.
For example, we have $\tau(\R_+^n,0,u) > 0$ iff $u$ is in the 
normal cone of $\R_+^n$ at $0$, that is, $u_i\le 0$ for all $i$. 
In this case, $\tau(\R_+^n,0,u) =\infty$.
The next lemma is a slight variation of a result by 
Federer~\cite{Federer_1959}. 

\begin{lemma}\label{lem:reach-conv}
Let~$W\subseteq E$ be a closed subset, $p\in W$, and 
$u\in E$ be a unit vector such that $\tau(W,p,u)$ is positive. 
Then we have $\tau(W,p) \le \tau(W,p,u)$.
\end{lemma}

\begin{proof}
The assertion is trivial if $\tau(W,p,u) = \infty$. 
So assume that $0 <\tau(W,p,u) < \infty$.  
Federer, in~\cite[Theorem~4.8(6)]{Federer_1959} states that 
under this assumption, the point 
$x :=p + \tau(W,p,u)u$ lies in the closure of~$\Delta_W$.
Therefore 
$\tau(W,p) \leq \|x-p\| = \tau(W,p,u)$.
\end{proof}

\begin{proof}[Proof of Theorem~\ref{thm:tau-boundary}] 
Let~$x\in \Delta_{W\cap V}$ and $p$ and~$q$ be distinct points 
in~$W\cap V$ such that~$d_{W\cap V}(x) = \|x-p\| = \|x-q\|$. 
It is sufficient to prove that 
\begin{equation}\label{eq:tautau}
   \|x-p\| \geq\min( \tau(W), \tau(W\cap \partial V)) ,
\end{equation}
since the assertion then follows by taking 
the infimum of~$\|x-p\|$ over~$x\in \Delta_{W\cap V}$.

If both~$p$ and~$q$ lie in~$\partial V$, 
then $x\in \Delta_{W\cap \partial V}$ and 
$\|x-p\| = d_{W\cap V}(x) = d_{W\cap \partial V}(x) 
\geq \tau(W\cap \partial V)$,
which implies \eqref{eq:tautau}. 

So we may assume that one of~$p$ and~$q$, say~$p$, 
does not lie on~$\partial V$, that is, $p$ is an interior point of~$V$.
Consider the unit vector $u:=\frac{x-p}{\|x-p\|}$
(note that $x\ne p$). 
We first observe that $\tau(W\cap V, p, u) \leq \|x-p\|$, 
because of the presence of the point~$q$, see 
Figure~\ref{fig:proof-reach}. 
Moreover, $\tau(W\cap V, p, u) > 0$ since $d_W(x)=\|x-p\| >0$. 
From this we can we deduce that $\tau(W, p, u) > 0$.
Indeed, 
the sets~$W$ and~$W\cap V$ coincide on a neighborhood of~$p$,
hence the distance functions~$d_W$ and~$d_{W\cap V}$ 
coincide for points on the segment $[p,x]$ that are sufficiently 
close to $p$. 
Using Lemma~\ref{lem:reach-conv}, we then obtain 
$$
    \tau(W) \le \tau(W,p) \le \tau(W,p,u)\leq \|x-p\| ,
$$
which shows \eqref{eq:tautau} and completes the proof.
\end{proof}

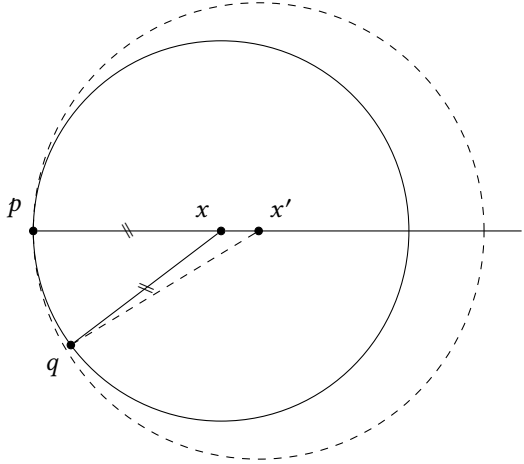
\begin{figure}[t]\centering
\begin{tikzpicture}[scale=.5,point/.style={draw,minimum size=0pt,
inner sep=1pt,circle,fill=black}]  
    \draw (0,0) node(p) [point,label=120:$p$] {};
    \draw (1,-3) node(q) [point,label=-120:$q$] {};
    \draw (5,0) node(x) [point,label=120:$x$] {};
    \draw (6,0) node(y) [point,label=60:$x'$] {};
    \draw (p) -- node[sloped,rotate=30] {\tiny $\parallel$} (x) -- (13,0);
    \draw (x) -- node[sloped,rotate=30] {\tiny $\parallel$}(q);
    \draw(x) circle [radius=5];
    \draw[dashed] (6,0) circle [radius=6];
    \draw[dashed] (y) -- (q);
\end{tikzpicture}
\caption{Illustration of the inequality~$\tau(W,p,u) \leq \|x-p\|$: 
A point~$x'$ beyond~$x$ on the half-line from~$p$ to~$x$ is 
closer to~$q$ than to~$p$.}
\label{fig:proof-reach}
\end{figure}

We can extend Theorem~\ref{thm:tau-boundary}  
to the intersections of several closed subsets. 

\begin{corollary}
\label{thm:tau-boundaries}
For closed subsets $V_1,\dotsc,V_s$ and~$W$ of~$E$ 
we have 
\[ 
   \tau(W\cap V_1 \cap \dotsb \cap V_s) \geq 
   \min_{I \subseteq \{1,\dotsc,s\}}
   \tau \Big( W \cap \bigcap_{i\in I} \partial V_i \Big) .
\]
\end{corollary}

\begin{proof}
The case $s=1$ is covered by Theorem~\ref{thm:tau-boundary}.
In general, we argue by induction on $s$, 
\begin{align*}
    \tau(W\cap V_1 &\cap \dotsb \cap V_{s+1}) \\
    &\geq \min\big(\tau(W \cap V_1 \cap \dotsb \cap V_s), 
    \tau( W \cap V_1 \cap \dotsb \cap V_s \cap \partial V_{s+1} )\big) \\
    &\geq \min \Big(  \min_{I \subseteq \{1,\dotsc,s\}}
    \tau \big( W \cap \bigcap_{i\in I} \partial V_i \big), 
    \min_{I \subseteq \{1,\dotsc,s\}}
      \tau \big( W \cap \partial V_{s+1} \cap 
    \bigcap_{i\in I} \partial V_i \big) \Big)\\
    &= \min_{I \subseteq \{1,\dotsc,s+1\}}
    \tau \Big( W \cap \bigcap_{i\in I} \partial V_i \Big) ,
\end{align*}
where we have applied Theorem~\ref{thm:tau-boundary}  
and twice the induction hypothesis.   
\end{proof}

We conclude with a relation between the reach of a subset of
the unit sphere $\IS(E) := \{x\in E \mid \|x\| = 1 \}$ 
and the reach of the cone over it.

\begin{lemma}\label{lem:tau-spherical}
Let $V\subseteq\IS(E)$ be closed and~$\widehat V=\R\cdot V$ be the closed cone in $E$ spanned
by~$V$. For any~$p\in V$, we have  
$\tau(V, p) \geq \min\{\tau(\widehat V, p), 1\}$.
\end{lemma}

\begin{proof}
We may assume that $E$ equals the span of $\widehat V$ because the reach of a subset
remains unchanged after restriction to a subspace that contains this subset. It follows from
this assumption that, for all $x\in E$,
$\pi_{\widehat V}(x) = 0$ if and only if~$x = 0$.

Elementary geometry shows that for all $x\in E$, whenever~$\pi_{\widehat V}(x)$ is
well-defined and not zero, then $\pi_V(x)$ is well-defined and
\begin{equation*}
    \pi_V(x) = \frac{\pi_{\widehat V}(x)}{\|\pi_{\widehat V}(x)\|}.
\end{equation*}
Now let $p\in V$. Recall that the reach~$\tau(V, p)$ is the supremum of all~$r> 0$ such
that~$\pi_V$ is well-defined on~$B(p,r)$.

Consider any $r<\min\{\tau(\widehat V, p), 1\}$ and let $x\in B(p,r)$. As $r<1$ we have
$x\neq 0$, and as $r<\tau(\widehat V, p)$, we have that $\pi_{\widehat{V}}(x)$ is well-defined
and not zero (as $x\neq0$). Which, as we noted above, implies that $\pi_V(x)$ is
well-defined. This shows that $\pi_V$ is well-defined on all of $B(p,r)$ for all
$r<\min\{\tau(\widehat V, p), 1\}$, from where the claim follows.
\end{proof}

\subsection{An extension of the Niyogi-Smale-Weinberger theorem}
\label{sec:an-extension-nsw}

Again, we work in a Euclidean vector space~$E$. 
By the \emph{(open) neighborhood} of radius~$r \geq 0$ around a 
nonempty set~$S\subseteq E$ we understand the set 
\begin{equation*}
    \mcU(S, r) \eqdef \big\{ p \in E \mid d_S(p) < r \big\}.
\end{equation*}

Niyogi, Smale and Weinberger~\cite[Prop.~7.1]{NiSmWe08} gave 
an answer to the following question: given a compact submanifold 
$S\subseteq
E$, a finite set $\mcX\subset E$ and $\e>0$, which conditions do we
need to ensure that $S$ is a deformation retract of $\mcU(\mcX,\e)$?

In what follows, we observe their arguments extend to any compact subsets
$S,\mcX$ provided $S$ has positive reach~$\tau(S)$. The proof is a variation of
the original proof. A much more general extension~\cite{ChazalCohen-SteinerLieutier_2009} of the
Niyogi-Smale-Weinberger theorem includes our, with slightly worse constants and,
naturally, at the cost of a more involved proof.\footnote{We thank Théo Lacombe
  and a referee for raising this point.}

The \emph{Hausdorff distance} between two nonempty closed 
subsets~$A,B\subseteq E$ is defined as 
\begin{equation*}
    d_H(A,B) \eqdef \max\Big( \sup_{a\in A} d_B(a), 
    \sup_{b \in B} d_A(b) \Big).
\end{equation*}

\begin{theorem}\label{thm:SNW}
Let $S$ and~$\mcX$ be nonempty compact subsets of~$E$.
The set~$S$ is a deformation retract of\/~$\mcU(\mcX,\epsilon)$
for any~$\epsilon$ such that
$ 3\, d_H(S, \mcX) < \epsilon < \tfrac{1}{2}\, \tau(S) $.
\end{theorem}

\begin{proof}
For any~$x\in \mcU(\mcX, \epsilon)$ we have
\[ 
  d(x, S) \leq d(x,\mcX) + d_H(\mcX, S) < \tfrac43 \epsilon < \tau(S), 
\]
hence $\mcU(\mcX,\epsilon) \subseteq T(S)$. This shows that 
the map 
\[ 
   \mcU(\mcX,\epsilon) \times [0,1] \to E,\quad 
   (x,t) \longmapsto (1-t) x + t\pi_S(x)
\]
is well-defined. The map is also continuous 
(Proposition~\ref{pro:piW}). 
It remains to prove that its image is included 
in~$\mcU(\mcX,\epsilon)$, that is,
for any $v\in \mcU(\mcX,\epsilon)$ the line segment~$[v,\pi_S(v)]$ 
is included
in~$\mcU(\mcX,\epsilon)$. The argument involves seven points 
depicted in Figure~\ref{fig:proof-NSW}.

\begin{figure}[t]\centering
\begin{tikzpicture}[scale=.5,point/.style={draw,minimum size=0pt,inner
      sep=1pt,circle,fill=black}]
    \draw[thick] (-7,0) -- node[above, at end] {$S$} (12,0);
    \draw
    (0,0) node(p) [point,label=-60:$p$] {} --
    (0,1) node(u) [point,label=180:$u$] {} --
    (0,5) node(v) [point,label=30:$v$] {} --
    (0,10) node(w) [point,label=60:$w$] {} --
    node[left,near end] {$\ell$} (0,12);
    \draw (4,4) node(x) [point,label=60:$x$] {};
    \draw[dashed,->] (x) -- node [below] {$\epsilon$} (u);
    \draw(x) circle [radius=5];
    \draw  (4,0) node(q) [point,label=60:$q$] {};
    \draw (-7,3) arc [radius=5,start angle=126.9,end angle=-60] ;
    \draw (-4,-1) node(y) [point,label=200:$y$] {};
    \draw[dashed,->] (y) -- node[below] {$\epsilon$} (0,-4);
    \draw[dashed,->] (y) -- node[below] {$\leq r$} (p);
    \draw[dashed,->] (x) -- node[right] {$\leq r$} (q);
    \draw[dashed,->] (x) -- node[above] {$< \epsilon$} (v);
    \draw[dotted] (w) -- (-5.2,10);
    \draw[dashed,<->] (-4.5,10) -- node[left,near start] {$6r$} (-4.5,0);
\end{tikzpicture}
\caption{Schematic view of the proof of Theorem~\ref{thm:SNW}}
\label{fig:proof-NSW}
\end{figure}
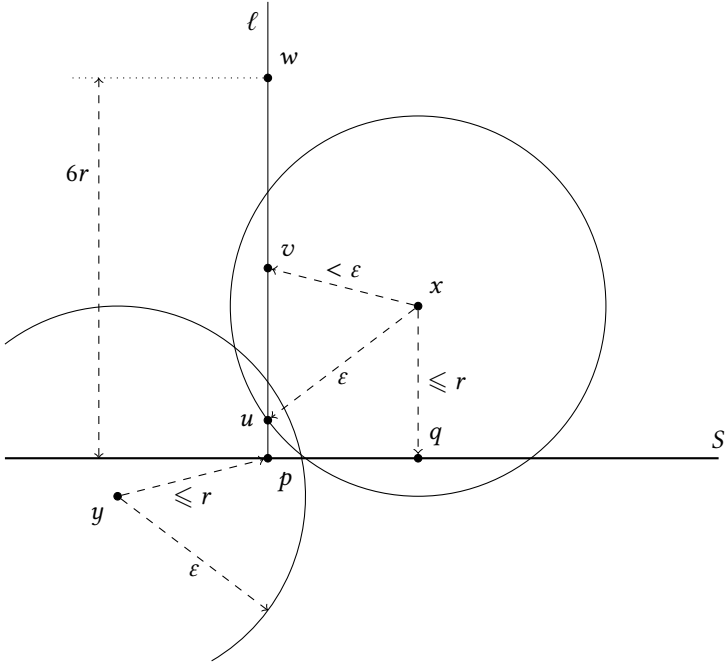

Let~$v \in \mcU(\mcX,\epsilon)$ and~$p\eqdef \pi_S(v)$. By definition,
there is some~$x\in\mcX$ such that~$\|v-x\|<\epsilon$. If~$\|p-x\| <
\epsilon$, then the line segment~$[v,p]$ is entirely included in the
ball of radius~$\epsilon$ around~$x$, which is a part
of~$\mcU(\mcX,\epsilon)$, and we are done. So we assume 
that~$\|p-x\|
\geq \epsilon$. Let~$u$ be the unique point in~$[v,p]$ such
that~$\|u-x\| = \epsilon$. The line segment~$[v,u)$ being included in
the ball $B(x,\e)\subseteq\mcU(\mcX,\epsilon)$, it only remains to
check that~$[u,p]$ is also included in~$\mcU(\mcX,\epsilon)$.

Let $r\eqdef \frac13 \epsilon$. Also, let~$\ell$ be the open half-line
starting from~$p$ and passing through~$v$ and~$w$ be the 
unique point
in~$\ell$ such that~$\|w-p\| = 6r$.  Our assumption states that $6r =
2\e< \tau(S)$. Also, as $p=\pi_S(v)$, we have
$\tau\big(S,p,\frac{w-p}{\|w-p\|}\big)>0$. By
Lemma~\ref{lem:reach-conv}, we obtain $\tau(S) \leq \tau(S,p) \leq
\tau\big(S,p,\frac{w-p}{\|w-p\|}\big)$, and therefore $6r <
\tau\big(S,p,\frac{w-p}{\|w-p\|}\big)$. This implies that $\pi_S(w)=p$
and $d_S(w) = \|w-p\| = 6r$.

Let~$q \eqdef \pi_S(x)$. We first note that~$\|x-q\| \leq r$ because
$d_S(x)\leq d_H(\mcX,S) < r$ by our assumption. Next we have
\begin{equation}\label{eq:5r}
 \|w-x\| \geq \|w-q\| - \|q-x\| \ge 5r,
\end{equation}
because~$\|w-q\| \geq d_S(w) = 6r$.

Since $d_H(\mcX,S) \leq r$, there is a point~$y\in \mcX$ such
that~$\|y-p\| \leq r$. To conclude the proof, it is enough to prove
that~$[u,p]\subseteq B(y,\e)$; since~$\|y-p\| \le r < \epsilon$, it is
sufficient to check that~$\|y-u\| <
\epsilon$. By the triangle inequality,
\begin{equation}\label{eq:1} 
   \|y-u\| \leq \|y-p\| +\|p-u\| \leq r + \|w-p\| - \|w-u\| 
     = 7r - \|w-u\|.  
\end{equation} 
Furthermore, the triangle $(xuv)$ has an acute angle at~$u$,
because~$u$ is on a sphere of center~$x$ and~$v$ lies inside this
sphere; the same holds true for the triangle~$(xuw)$ because~$u$, $v$
and~$w$ are on the same line, in this order
(cf. Figure~\ref{fig:proof-NSW}). It follows that $\|w-u\|^2 \geq
\|w-x\|^2 -\|x-u\|^2 > (5r)^2 - \epsilon^2$, where we 
used~\eqref{eq:5r} for the
second inequality. Therefore, with~\eqref{eq:1},
\begin{equation}\label{eq:2} 
   \|y-u\| < 7r - \sqrt{25 r^2 - \epsilon^2} 
  =3r=  \epsilon,
\end{equation} 
which concludes the proof.
\end{proof}

\section{Shub--Smale theory and extensions}
\label{sec:shub-smale-theory}

We recall the definition and basic properties of quantities 
($\alpha$, $\beta$ and~$\gamma$ numbers) introduced by 
Shub and Smale to study the complexity of
numerical methods for solving polynomial systems. 
We prove two results. First, an analogue of the $\alpha$-Theorem 
for the continuous Newton method, where a
continuous Newton's flow replaces the 
discrete sequence obtained with Newton's iteration. Second, 
an inequality relating the reach and the
$\gamma$-number. It strengthens and simplifies a result of
Cucker, Krick and Shub~\cite{CuckerKrickShub_2018},  
for it is pointwise whereas the latter is only global.

\subsection{Measures of proximity} 

Let~$E$ be a Euclidean space and~$F\colon E\to \R^m$ be an analytic
map such that $m\leq\dim E$. It is well-known that, under certain
conditions, the (Moore-Penrose) Newton iteration with initial point
$x_0$, given by
\begin{equation}\label{eq:Newton}
     x_{k+1}:=x_k-\Diff F(x_k)^\dagger F(x_k)
\end{equation}
is well-defined for all $k\ge 0$ and converges quadratically fast to a
zero $\zeta$ of $F$. In this case, we say that $x_0$ is an {\em
approximate zero} of $F$ with {\em associated zero} $\zeta$ (see
\cite[Def.~15.1]{Condition} for the formal definition). Here
$\Diff F(x)^\dagger\colon \R^{m} \to E$ is the {\em Moore-Penrose inverse}
of the full-rank matrix $\Diff F(x)$ (we say that~\eqref{eq:Newton} is
undefined if it is not of full rank). 

An obvious question is whether we can, for a given $x_0$, ensure the
convergence of Newton's iteration. That is, whether we can check that
$x_0$ is an approximate zero of $f$. An answer to this question was
provided by Smale~\cite{Smale_1986} for the zero-dimensional case
($m=\dim E$) and extended by Shub and 
Smale~\cite{Bez4} to
the underdetermined case ($m<\dim E$). They defined the quantities
(all norms are the spectral one)
\begin{equation}\label{eq:def-abc}
\begin{aligned}
\gamma(F, x) & \eqdef \sup_{k\geq 2} 
\left\| \tfrac{1}{k!} \Diff F(x)^\dagger \Diff^k F(x) \right\|^{\frac{1}{k-1}}, \\
\beta(F, x) &\eqdef \left\| \Diff F(x)^\dagger F(x)  \right\| ,\\[2pt]
\alpha(F, x) &\eqdef \gamma(F, x) \beta(F, x) ,
\end{aligned}
\end{equation}
and proved that there exists a universal constant 
$\alpha_\bullet$, around $\frac18$, such that if 
$\alpha(F,x_0)<\alpha_\bullet$ then $x_0$ is an approximate zero 
of $F$. The quantities $\beta$ and $\gamma$ are not without 
meaning themselves. Clearly, $\beta(F,x)$ is the length of the 
Newton step at $x$. Also, in the zero-dimensional case, 
Smale's $\gamma$-Theorem shows that, for a zero $\zeta$ of 
$F$, all points on the ball around $\zeta$ 
with radius $\frac{3-\sqrt{7}}{2\gamma(F,\zeta)}$ 
are approximate zeros of $F$.

\subsection{Continuous \texorpdfstring{$\alpha$}{alpha}-theory}
\label{sec:cont-alpha-theory}

While the Shub--Smale theory focuses primarily on the discrete
iteration~\eqref{eq:Newton}, the numbers~$\alpha$ and~$\beta$ also
give quantitative information on the convergence of the continuous
analogue of Newton's iteration. To the best of our knowledge, this has
never been highlighted before.

We consider again a point~$x_0$ in~$E$ such that $\Diff F(x_0)$ is
surjective. We define~$x_t$ with the following system of ordinary
differential equations, for $t$ in the maximal domain of solution
containing~$0$,
\begin{equation}\label{eq:diffeqx}
   \dot x_t = - \Diff F(x_t)^\dagger F(x_t),
\end{equation}
where~$\dot x_t$ denotes~$\frac{\ud}{\ud t} x_t$. 
We also denote~$\alpha_t\eqdef \alpha(F, x_ t)$ 
and~$\beta_t$ and~$\gamma_t$ accordingly. 
It may be that~$\gamma_t$ (and thus~$\alpha_t$) is not differentiable
everywhere. However, $\gamma_t$ is at least locally Lipschitz continuous
(cf. Lemma~\ref{lem:cont-alpha}),
which implies absolute continuity and, in turn, that~$\gamma_t$ is
differentiable almost everywhere and 
that~$\gamma_t - \gamma_0 = \int_0^t \dot \gamma_t \ud t$ \cite[IX\S4]{Natanson_1955}. 
This regularity is good enough for our
purposes. In all our arguments below, at a point~$t$ 
where~$\gamma_t$ is not differentiable, an
inequality like~$\dot \gamma_t \leq 5 \gamma_t^2 \beta_t$ actually
means
\[ 
  \limsup_{\epsilon\to 0} \frac{\gamma_{t+\epsilon}-\gamma_t}
  {\epsilon} \leq 5\gamma_t^2 \beta_t. 
\]

The {\em domain of definition}~$\Omega$ of the differential equation 
is the open set of all~$x\in E$ such that~$\Diff F(x)$ is surjective.

\begin{theorem}\label{thm:cont-alpha}
  If~$\alpha_0 < \frac{1}{13}$,
  then~$x_t$ is defined for all~$t \geq 0$ and
  \begin{enumerate}[(i)]
  \item \label{item:2} $F(x_t) = F(x_0) e^{-t}$;
  \item \label{item:5} $\|x_t-x_0\| \leq 2 \beta_0 (1-e^{-t})$;
  \item \label{item:6} $x_t$ converges when~$t\to\infty$.
  \end{enumerate}
\end{theorem}

\begin{lemma}\label{lem:cont-alpha}
  For all~$t\geq 0$ where $x_t$ is defined, we have 
  \begin{enumerate*}[(i)]
  \item \label{item:1} $\|\dot x_t\| = \beta_t$,
  \item \label{item:3} $\dot \gamma_t \leq 5 \gamma_t^2 \beta_t$,
  \item \label{loc-Lip} $t\mapsto\gamma_t$ is locally Lipschitz, 
  \item \label{it:dotbeta} 
    $\dot \beta_t \leq - \beta_t + 3 \gamma_t \beta_t^2$, and
  \item \label{item:4} $\dot \alpha_t \leq - \alpha_t + 8 \alpha_t^2$.
  \end{enumerate*}
\end{lemma}

\begin{proof}
\ref{item:1} It follows from~\eqref{eq:diffeqx}.   

\ref{item:3} Let~$y = x_{t+\epsilon}$ for some small positive~$\epsilon$.
Let~$u = \|x-y\| \gamma_t = \gamma_t\|\dot x_t\| \epsilon  +
\mcO(\epsilon^2)$.
By \cite[Lemme~131]{Dedieu_2006},
\begin{align*}
  \gamma(F, y) - \gamma(F, x_t)
  &\leq \left( \frac{1}{(1-u)(1-4u+2u^2)} -1 \right) \gamma(F, x_t) \\
  &= 5u \gamma_t + \mcO(\epsilon^2),
\end{align*}
so that 
$\gamma_{t+\epsilon} - \gamma_t \leq 5 \gamma_t^2 \|\dot x_t\|\epsilon + \mcO(\epsilon^2)$. 
With the equality~$\|\dot x_t\| = \beta_t$, 
this gives the claim 
when taking the limit for $\epsilon\to 0$. 
  
\ref{loc-Lip} This follows from \ref{item:3}. 

\ref{it:dotbeta}
Let $A_t \eqdef \Diff F(x_t)$, $B_t \eqdef A_t^\dagger F(x_t)$, 
and $P_t := \mop{id}_E - A_t^\dagger A_t$ (the orthogonal projection on $\ker A_t$). 
Then $\dot{A}_t = \Diff^2F(x_t)(\dot{x}_t)$ and (we drop the index $t$)
\cite[Thm.~4.3]{GolubPereyra_1973}
\[ 
   \tfrac{\ud}{\ud t} A^\dagger = -A^\dagger \dot A A^\dagger 
 + P (A^\dagger\dot A)^T A^\dagger.
\]
This formula derives from the equality~$A^\dagger = A^T (A A^T)^{-1}$
which holds because~$A$ is surjective.

Using that~$\Diff F(x) \Diff F(x)^\dagger = \mop{id}_E$ and given the
differential equation~\eqref{eq:diffeqx} for~$x_t$, we check that, 
for any~$t\in I$, 
\begin{equation}\label{eq:12}
  \tfrac{\ud}{\ud t} F(x_t) = \Diff F(x_t)(\dot x_t) = - \Diff F(x_t) \Diff F(x_t)^\dagger F(x_t) = - F(x_t).
\end{equation}
Using this equality we deduce that
\begin{align*}
  \dot{B} &= A^\dagger \tfrac{\ud}{\ud t} F(x) + \left( \tfrac{\ud}{\ud t} A^\dagger \right) F(x)\\
          &= - A^\dagger F(x) +  \left( -A^\dagger \dot A A^\dagger + P (A^\dagger \dot A)^T A^\dagger \right) F(x) \\
          &= -{B} +\big[ -\Diff F(x)^\dagger \Diff^2 F(x)(\dot{x})  
     + P (\Diff F(x)^\dagger \Diff^2F(x)(\dot{x}))^T \big] B .
\end{align*}
Let~$C$ denote the operator inside the square brackets.
Since  the two terms in~$C$ have orthogonal images, we easily 
obtain 
\[ 
    \|C\| \leq \sqrt{2}\, \|\Diff F(x)^\dagger \Diff^2F(x)(\dot x) \| 
   \leq 2\sqrt{2} \gamma \|\dot x\| < 3 \gamma \beta. 
\]
Since~$\dot\beta_t = \frac{1}{\beta_t} \langle \dot B_t, B_t \rangle$, 
we obtain item~\ref{it:dotbeta}.

\ref{item:4} From the previous inequalities,
\[ 
  \dot \alpha_t = \dot \gamma_t \beta_t + \gamma_t \dot\beta_t 
  \leq 8 \gamma_t^2 \beta_t^2 - \gamma_t \beta_t, 
\]
which is exactly the claim.
\end{proof}

\begin{proof}[Proof of Theorem~\ref{thm:cont-alpha}]

Let~$I = [0,\tau)$ be the maximum domain of solution of the
differential equation~\eqref{eq:diffeqx} with the fixed initial
condition~$x_0$. We will shortly see that~$\tau=\infty$.

By Equation~\eqref{eq:12} we have that, for any~$t\in I$, 
$\tfrac{\ud}{\ud t} F(x_t)=- F(x_t)$. This gives the first claim.

After some calculation, using Lemma~\ref{lem:cont-alpha}\ref{item:4}, 
we check
that $\frac{\ud}{\ud t}\frac{e^{-t}}{\alpha_t} \geq -8 e^{-t}$. It follows
that for~$t\in I$,
\begin{equation}\label{eq:7}
 \alpha_t \leq \frac{\alpha_0 e^{-t}}{1-8\alpha_0}.
\end{equation}
After some calculation, using 
Lemma~\ref{lem:cont-alpha}\ref{it:dotbeta}, we
obtain that~$\frac{\ud}{\ud t}\log \beta_t \leq -1 + 3\alpha_t$, and 
therefore, for $t\in I$,
\begin{equation}\label{eq:b-bound}
  \beta_t \leq \beta_0 \exp \left( \tfrac{3\alpha_0}{1-8\alpha_0} 
    - t\right) \leq 2 \beta_0 e^{-t}, 
\end{equation}
where we used that~$\alpha_0 \leq \frac{1}{13}$ for the second
inequality.  Using Lemma~\ref{lem:cont-alpha}\ref{item:1}, we compute
that~$-\frac{\ud}{\ud t} \frac{1}{\gamma_t} \leq 5 \beta_t$, and it
follows with \eqref{eq:7} that for~$t\in I$,
\begin{equation}
  \label{eq:g-bound}
  \gamma_t \leq \frac{\gamma_0}{1- 10 \beta_0 \gamma_0} 
  \leq \tfrac{13}{3} \gamma_0.
\end{equation}
By Inequality~\eqref{eq:b-bound} and Lemma~\ref{lem:cont-alpha}(i),
$\|\dot x_t\| = \beta_t$ is bounded for $t\in I$. Therefore, if the
interval~$I=[0,\tau)$ is bounded, then $x_t$ approaches, as
$t\to\tau$, 
a point $y$ in 
the complement of~$\Omega$, the domain of definition of
the differential equation
\cite[IV.5, Th.~2]{Bourbaki_FRV}. 
Therefore,
$\gamma_0 \|y-x\| \le 2 \beta_0 \gamma_0 = 2\alpha_0 
< 1 -\tfrac12 \sqrt{2}$, 
and~\cite[Lemme~123]{Dedieu_2006} 
implies that $\Diff F(y)$ is surjective, which contradicts $y\not\in\Omega$. 
We have thus shown that~$\tau=\infty$.

Next, with Lemma~\ref{lem:cont-alpha}\ref{item:1}, it follows that,
for~$t\in I$,
\[ 
  \|x_t-x_0\| \leq \int_0^t \|\dot x_s\| \ud s \leq 2 \beta_0  (1 - e^{-t}). 
\]
which is the second claim.
Similarly, Equation~\eqref{eq:b-bound} shows that the 
integral~$\int_0^\infty \dot x_s\ud s$ is absolutely convergent, 
therefore~$x_t$ has a limit when~$t\to\infty$.
\end{proof}

\subsection{An inequality relating the reach and the 
\texorpdfstring{$\gamma$}{gamma}-number}
\label{sec:proof-tau-gamma}

We keep assuming that $E$ is a Euclidean space 
and~$F\colon E\to\R^m$
an analytic map, with $m\le \dim E$. We will prove the following local
inequality, which is a refinement of a result first proved by Cucker,
Krick and Shub~\cite{CuckerKrickShub_2018}.  The proof is also much
simpler.

\begin{theorem}\label{thm:tau-gamma-local}
Let~$\mcM\subseteq E$ be the zero set of the 
analytic map~$F\colon E\to\R^m$ and $p\in\mcM$. 
Then we have $\tau(\mcM,p) \gamma(F,p) \geq \frac{1}{14}$
if $\gamma(F,p) < \infty$. 
\end{theorem}

For~$p\in E$ such that $\rank \Diff F(p)=m$, let $\pi_p\colon E\to E$
denote the orthogonal projection onto the kernel of~$\Diff F(p)$, that is,
$\pi_p = \mop{id}_E - \Diff F(p)^\dagger \Diff F(p)$. Note that if $\gamma(F,p)
< \infty$ then locally around $p$, $\mcM$ is a smooth manifold and
$\ker \Diff F(p)$ is the tangent space $T_p\mcM$ at~$p$.

\begin{proposition}\label{lem:Np-lip}
The derivative of the rational map 
$\pi\colon E \to \mop{End}(E),\, p \mapsto \pi_p$
at ~$p \in E$ has an operator norm bounded by~$2 \gamma(F,p)$. 
Here $\mop{End}(E)$ is endowed with the spectral norm.
\end{proposition}

\begin{proof}
Let~$p\in \mcM$. The derivative of~$\pi$ at~$p$,  
$\Diff\pi(p):E\to\mop{End}(E)$, evaluated at~$\dot p \in E$ 
yields~\cite[Cor.~4.2]{GolubPereyra_1973},
\begin{equation*}
    \Diff \pi(p)(\dot p) = 
    - \Diff F(p)^\dagger \cdot \Diff^2F(p)(\dot p)\cdot\pi_p 
    - ( \Diff F(p)^\dagger\cdot \Diff^2F(p)(\dot p) \cdot \pi_p)^T.
\end{equation*}
Since $\left\| \tfrac{1}{2} \Diff F(p)^\dagger \Diff^2 F(p) \right\| 
\le \gamma(F,p)$,
by the definition~\eqref{eq:def-abc} of~$\gamma(F,p)$, 
it follows that $\left\|D \pi(p) \right\| \le 4 \gamma(F,p)$.
We obtain the sharper bound $2\gamma(F,p)$ by observing 
that~$\|A + A^T\| =\|A\|$ for any map~$A \in \mop{End}(E)$ 
such that $A^2=0$, which holds for
$A=\Diff F(p)^\dagger \, \Diff^2F(p)(\dot p)\,\pi_p$.
\end{proof}

It is worth noting that the derivative $\Diff\pi(p)$ is an incarnation of
the second fundamental form $B_p\colon T_p\mcM \times T_p \mcM\to
T_p\mcM^\perp$ of $\mcM$ at $p$ and one can see that $\|\Diff\pi(p)\| =
\|B_p\|$.  Proposition~\ref{lem:Np-lip} means that the norm of the
second fundamental form of $\mcM$ at~$p$, a classical measure of
curvature in differential geometry, is bounded by $2\gamma(F,p)$. This
is related to \cite[Prop.~6.1]{NiSmWe08}, where this norm is upper
bounded by $1/\tau(\mcM)$.

\begin{proof}[Proof of Theorem~\ref{thm:tau-gamma-local}]
We fix~$p\in \mcM$ such that $\gamma(p) < \infty$.  Since~$\tau(p) =
\inf_{u\in \Delta_\mcM} \|u-p\|$, it is enough to prove that
$\gamma(p) \|u-p\| \geq \frac{1}{14}$ for any given~$u\in
\Delta_{\mcM}$.  To shorten notation, we write~$\gamma(p)$
for~$\gamma(F,p)$.  

Let $u\in\Delta_\mcM$. By definition, there exist distinct points~$x$
and~$y$ in~$\mcM$ such that $d_{\mcM}(u) = \|u-x\| = \|u-y\|$. Using
the triangle inequality (three times) we see that 
$$
   \max(\|x-y\|,\|p-x\|,\|p-y\|) \le 2 \|u-p\| . 
$$ 
Therefore, denoting
$$ 
   \eta := \gamma(p) \max(\|x-y\|,\|p-x\|,\|p-y\|), 
$$ 
we obtain $\gamma(p) \|u-p\| \geq \frac12 \eta$. 
If~$\eta \geq \frac17$, then we
are done, so we can assume that~$\eta <\frac17$.

Let~$B \subseteq E$ be the ball of center~$p$ and
radius~$\eta/\gamma(p)$ (in particular $x,y\in B$). Since~$\eta \leq
\frac17$, $\gamma$ is bounded on~$B$ by $K \gamma(p)$, where~$K \eqdef
\frac{1}{(1-\eta)(1-4\eta+2\eta^2)}$ \cite[Lemme~131]{Dedieu_2006}. In
particular, $\gamma(x)$ and~$\gamma(y)$ are finite, so that $x$ and
$y$ are regular points of $\mcM$.

We now give a lower and an upper bound for~$\|\pi_x(u-y)\|$. Let~$y'=
x + \pi_x(y-x)$. By Lemma~\ref{le:ortho}, the vector~$u-x$ is normal
to~$\mcM$ at~$x$, that is $\pi_x(u-x) = 0$; thus~$\pi_x(u-y) = x-y'$,
and then we have the lower bound 
\begin{equation}\label{eq:lobo}
\|x-y\| -\|y-y'\| \leq \|\pi_x(u-y)\| . 
\end{equation}
Similarly, the vector~$u-y$ is normal to~$\mcM$ at $y$,
that is~$\pi_y(u-y) =0$; hence the upper bound \[ \|\pi_x(u-y) \| =
\|\pi_x (u-y) - \pi_y(u-y)\| \leq \|\pi_x-\pi_y\|\cdot \|u-y\|. \] 
Combined with \eqref{eq:lobo}, we obtain
\begin{equation} \label{eq:main-ineq-tau-gamma} 
 \|x-y\| - \|y-y'\| \le \|\pi_x-\pi_y\|\|u-y\| .  
\end{equation}

Further, we aim at bounding~$\|y-y'\|$.  By definition of $y'$, using
that~$F(x)=F(y) = 0$, and expanding $F(y)$ into a power series at~$x$,
we can write 
\begin{align*} 
y - y' &= \Diff F(x)^\dagger \Diff F(x)(y-x) -\Diff F(x)^\dagger F(y) \\ 
&= \Diff F(x)^\dagger \Diff F(x)(y-x) - \Diff F(x)^\dagger
\sum_{k\geq 0} \tfrac{1}{k!} \Diff^k F(x)(y-x,\dotsc,y-x) \\ 
&= -\sum_{k\geq 2} \tfrac{1}{k!} \Diff F(x)^\dagger \Diff^k F(x)(y-x,\dotsc,y-x).
\end{align*} 
Hence 
\begin{equation}\label{eq:9} 
\begin{aligned}
 \|y-y'\| &\leq \sum_{k\geq 2} \big\|\tfrac{1}{k!} \Diff F(x)^\dagger \Diff^k
 F(x)\big\| \|y-x\|^k \\ 
 & \leq \|y-x\| \sum_{k\geq 2}\big(\gamma(x)\|y-x\|\big)^{k-1}\\ 
&=\frac{\gamma(x)\|x-y\|^2}{1-\gamma(x)\|x-y\|} \leq
 \frac{K\eta}{1-K\eta} \|x-y\|,  
\end{aligned} 
\end{equation}
the last inequality following from 
$\gamma(x) \|x-y\| \le K \gamma(p) \|x-y\| \le K \eta$ and 
the monotonicity of the function $t\mapsto t/(1-t)$.

Lastly, we bound~$\|\pi_x-\pi_y\|$.  By Proposition~\ref{lem:Np-lip},
we can upper bound 
\begin{equation}\label{eq:gaga}
 \|\pi_x - \pi_y\| \le \sup_{z\in [x,y]} \| \Diff\pi(z) \| \cdot \|x-y\| \le 
  \sup_{z\in [x,y]} 2 \gamma(z)  \cdot \|x-y\| \le 2 K\gamma(p) \|x-y\| .
\end{equation}
Combining \eqref{eq:main-ineq-tau-gamma}, \eqref{eq:9}
and~\eqref{eq:gaga}, we obtain 
$$ 
 \Big(1 - \frac{\eta K}{1 - \eta K}\Big) \|x-y\| \le 
 2 K \gamma(p) \|x-y\|\, \|u-y\|.  
$$ 
Dividing by the nonzero $\|x - y\|$ and noting $\|u-y\| \le \|u -p \|$, 
this implies
$$ 
  \frac{1}{14} \le \frac{1}{2K}\Big(1 - \frac{\eta K}{1 - \eta K}\Big) 
  \le \gamma(p) \|u-p\| , 
$$ 
where the left-hand inequality is easily checked numerically.
\end{proof}

\section{Condition number of semialgebraic systems}

We focus now on semialgebraic sets, and more specifically, on
{\em spherical} semialgebraic sets $S(F,G)$ given by homogeneous
semialgebraic systems $(F,G)$; cf.~\eqref{eq:hsa}. We do 
so since, eventually (see \S\ref{sec:affine} below), we will reduce 
the computation of the homology of semialgebraic sets and its 
complexity analysis to the same tasks for the spherical case.
We define a condition 
number $\kappasa$ for homogeneous semialgebraic systems 
and relate it to three different measures of conditioning: 
the distance to the closest ill-posed
system in the space of semialgebraic systems
(Theorem~\ref{thm:kappasa-dist}); the reach of 
the set $S(F,G)$ (Theorem~\ref{thm:tau-kappa}); and the sensitivity
of $S(F,G)$ to small relaxations of the equalities and 
inequalities of the system~$(F, G)$
(Theorem~\ref{thm:approx}). We also bound the degree of the
hypersurface of ill-posed systems
(Proposition~\ref{prop:discriminant}). We finally give a notion of
condition number for affine semialgebraic systems that is based 
on the one for the homogeneous case.

\subsection{Measures of condition}
\label{sec:cond-semi-sets}

\subsubsection{The \texorpdfstring{$\mu$}{mu} numbers}
\label{sec:number-mu}

To a degree pattern $\bfd=(d_1,\ldots,d_q)$ we associate 
the linear space $\Hd[q]$ of the polynomial systems 
$F=(f_1,\ldots,f_q)$ where $f_i\in\R[X_0,X_1,\ldots,X_n]$ 
is homogeneous of degree $d_i$.  Let~$D = \max_{1\leq i\leq q} d_i$.
We endow~$\Hd[q]$ with a Euclidean inner product, 
the \emph{Weyl inner product}, defined as follows.
For homogeneous polynomials 
$h=\sum_{|\bfa|= d}h_{\bfa} X^\bfa$ and 
$h'=\sum_{|\bfa|= d}h_{\bfa}' X^\bfa$ in~$\R[X_0,\dotsc,X_n]$,
where we write $\bfa=(a_0,\dotsc,a_n)\in\N^{n+1}$ and 
$|\bfa| \eqdef a_0+\dotsb+a_n$,
we define 
$$
  \langle h,h'\rangle\eqdef \sum_{|\bfa| = d} {d\choose \bfa}^{-1} 
  h_{\bfa} h_{\bfa}',
$$
where 
${d\choose \bfa} \eqdef\frac{d!}{a_0! a_1!\ldots a_n!}$ 
is the multinomial coefficient.
For any~$q$-tuples of homogeneous polynomials $F,F'\in\Hd[q]$
with degree pattern $\bd$, say
$F=(f_1,\ldots,f_q)$ and  $F'=(f_1',\ldots,f_q')$, 
we define 
$$
   \langle F,F'\rangle\eqdef\sum_{j=1}^q \langle f_j,f_j'\rangle.
$$
In other words, the Weyl inner product is a dot product with 
respect to a 
specifically weighted monomial basis. Its \emph{raison d'être} 
is the fact that it is invariant under orthogonal transformations 
of the homogeneous variables $(X_0,\ldots,X_n)$. That is, 
that for any orthogonal transformation 
$u:\R^{n+1}\to\R^{n+1}$ and any $F\in\Hd[q]$, we have 
$\|F\| = \| F\circ u \|$.
In all of what follows, all occurrences of norms in spaces 
$\Hd[q]$ refer to the norm induced by the Weyl inner product. 

For a point $x\in\R^{n+1}$ and a system~$F\in\Hd[q]$, let $\Diff F(x)$ denote the
derivative of $F$ at $x$, which is a linear map~$\R^{n+1}\to\R^q$. We also
define the diagonal normalization matrix
\begin{equation*}
\Delta \eqdef
\begin{pmatrix}
  \sqrt{d_1} & & \\
  & \ddots & \\
  & & \sqrt{d_q}
\end{pmatrix}.
\end{equation*}
The \emph{condition number} $\mun(F,x)$ of~$F\in \Hd[q]$ at~$x\in\IS^n$ has been
well studied~\cite{Bez1,Bez2,Bez3,Bez4,Bez5}, see also \cite{Condition}. We define
it as $\infty$ when the derivative $\Diff F(x)$ of $F$ at $x$ is not surjective, and
otherwise as
\begin{equation}\label{eq:munorm}
   \mun(F,x) \eqdef\|F\| \left\|\Diff F(x)^\dagger \Delta\right\|,
\end{equation}
where the norm $\|\Diff F(x)^\dagger \Delta\|$
is the spectral norm.
We also define the following variant of $\mun$, more specific 
to homogeneous systems, 
\begin{equation*}
  \mup(F, x) \eqdef \mun(F|_{\mcT_x},x) 
  = \|F\| \left\| \Diff F(x)|_{T_x}^\dagger \Delta \right\|,
\end{equation*}
where~$T_x = \left\{ x \right\}^\perp$ and~$\mcT_x \eqdef x + T_x$. 
(The number~$\mun(F|_{\mcT_x},x)$ is well-defined after 
identifying~$\mcT_x$ with~$\R^n$.) 

The following inequality is a useful result from the Shub--Smale
theory~\cite[Lemma~2.1(b)]{Bez4}.

\begin{proposition}\label{prop:gamma-mu}
Let $F\in\Hd[q]$ and $x\in\R^{n+1}$ be a zero of $F$. Then
\begin{equation}\tag*{\qed}
 \gamma(F,x) \leq  \tfrac12 D^{\frac32} \mun(F, x) .
\end{equation}
\end{proposition}

\subsubsection{The \texorpdfstring{$\kappa$}{kappa} number}
\label{sec:number-kappa}

The numbers $\mun(F,x)$ and $\mup(F,x)$ measure the sensitivity  
of the zero $x$ of $F$ when $F$ is slightly perturbed. They are 
consequently useful at a zero, or near a zero, of the system~$F$. 
To deal with points in $\IS^n$ far away 
from the zeros of $F$, in particular to understand how much $F$ 
needs to be perturbed to make such a point a zero,  
a more global notion of conditioning is needed. 
The following is (modulo replacing 
$\mun$ by $\mup$) the condition measure introduced 
in~\cite{Cucker99b} (see also~\cite[\S19]{Condition} 
and~\cite{CKMW1,CKMW2}). 

\begin{definition}\label{def:kappa}
The {\em real homogeneous condition number} of $F\in\Hd[q]$ at 
$x\in\IS^n$ is 
$$
 \kappa(F,x) \eqdef 
 \left( \frac{1}{\mup(F,x)^{2}} + \frac{\|F(x)\|^2}{\|F\|^2} \right)^{-1/2},
$$
where we use the conventions $\infty^{-1} :=0$, $0^{-1}:=\infty$, 
and $\kappa(0,x) := \infty$. 
We further define 
$\kappa(F) \eqdef\displaystyle \max_{x\in\IS^n} \kappa(F,x)$.
\end{definition}

If~$q > n$ (that is, if the system~$F$ is overdetermined) 
then $\Diff F(x)|_{T_x}$ cannot be surjective 
and~$\kappa(F,x) = \frac{\|F\|}{\|F(x)\|}$ for all $x\in\IS^n$. Thus, 
$\kappa(F) < \infty$ if and only if $F$ has no zeros in $\IS^n$. 

The special case $F(x)=0$ is worth highlighting.
\begin{lemma}\label{le:mu-proj-norm}
For any $F\in\Hd[q]$ and $x\in \IS^n$, if~$F(x)=0$, then
\[ 
   \kappa(F,x) =\mup(F,x) = \mun(F,x). 
\]
\end{lemma}

\begin{proof}
The first equality follows from the definition of~$\kappa$. For 
the second, recall that the pseudo-inverse~$\Diff F(x)^\dagger$ 
is the inverse of~$\Diff F(x)$ restricted as a
map~$(\ker \Diff F(x))^\perp \to \R^q$. If~$F(x)=0$, then~$\Diff F(x)(x) = 0$, 
by homogeneity, therefore the orthogonal complement of the kernel 
of~$\Diff F(x)$ is included in~$T_x$. It follows 
that~$\Diff F(x)|_{T_x}^\dagger = \Diff F(x)^\dagger$ and
then $\mup(F,x)=\mun(F,x)$.
\end{proof}

For $x\in \IS^n$, let~$\Sigma_x$ be the set of all~$F\in \Hd[q]$ 
such that~$\kappa(F,x) = \infty$, that is~$F(x) = 0$ 
and~$\Diff F(x)|_{T_x}$ is not surjective.
The {\em set of ill-posed algebraic systems} is defined as 
$\Sigma \eqdef \bigcup_{x\in \IS^n}\Sigma_x$. 
It is the set of all~$F\in \Hd[q]$ such that~$\kappa(F) =\infty$.
We have 
\begin{equation}\label{def:Sigma} 
 \Sigma = \big\{  F\in \Hd[q] \mid \exists x\in\IS^n\ F(x) = 0 
\text{ and $\Diff F(x)|_{T_x}$ is not surjective} \big\} .
\end{equation}
The set~$\Sigma$ is semialgebraic 
and invariant under scaling of each of the $q$ components.
Note that in the case $q>n$, the set $\Sigma_x$ just consists of 
the $F\in \Hd[q]$ such $F(x) = 0$, and $\Sigma$ equals the set of 
$F\in \Hd[q]$ that possess a real zero in $\IS^n$.  

\begin{theorem}\label{thm:kappa-dist}
For any nonzero~$F\in \Hd[q]$ and any~$x\in\IS^n$,
\begin{equation*}
    \kappa(F,x) = \frac{\|F\|}{d(F,\Sigma_x)} 
    \quad\text{and}\quad \kappa(F) = \frac{\|F\|}{d(F, \Sigma)} ,
  \end{equation*}
where the distance $d(F,\cdot)$ is defined via the norm 
induced by the Weyl inner product. 
\end{theorem}

\begin{proof}
The assertion is obvious in the case $q>n$. 
We therefore assume $q\le n$. The special case $q=n$ is 
Prop.~19.6 in~\cite{Condition}. 
One can check that the same proof works in the case $q\le n$. 
\end{proof}

\begin{corollary}\label{coro:kappa-geq-1}
For any~$F\in\Hd[q]$ and any~$x\in\IS^n$, $\kappa(F,x) \geq 1$.
\end{corollary}

\begin{proof}
Since~$0\in\Sigma_x$, this follows directly from 
Theorem~\ref{thm:kappa-dist}.
\end{proof}

\begin{remark}\label{rem:2kappas}
Proposition~6.1 in~\cite{CuckerKrickShub_2018} shows that for 
the condition 
number $\kappan(F,x)$ defined as in Definition~\ref{def:kappa} above, 
but with $\mup$ replaced by $\mun$, we have 
$$
      \frac{\|F\|}{\sqrt{2}\,d(F,\Sigma_x)} \leq\kappan(F,x)
      \leq\frac{\|F\|}{d(F,\Sigma_x)}. 
$$
This shows that there is no essential difference between 
$\kappa$ and $\kappan$: they are the same up to 
a factor of at most $\sqrt{2}$. 
It also shows that, for all $x\in\IS^n$, $\mun(F,x)\leq \mup(F,x)$. 
So the bound in Proposition~\ref{prop:gamma-mu} holds 
with $\mup(F,x)$ as well.

However, a bound on $\mup$ in terms of $\mun$ is not possible. 
Indeed, take $f_1(x,y,z) := x+y$ and $f_2(x,y,z) := y^2+z^2+xy$. 
Further, take $e_0:=(1,0,0)$. Then
$$
 \Diff F(e_0) = \begin{bmatrix}1 & 1 & 0 \\ 0 & 1 & 0 \end{bmatrix}, \quad
 \Diff F(e_0)_{\mid e_0^\perp} = \begin{bmatrix}1 &  0 \\ 1 & 0 
  \end{bmatrix},
$$
where the left-hand matrix is of full rank, but the right-hand matrix 
is rank deficient.
Hence $\mun(F,e_0) < \infty$, but 
$\mup(F,e_0) = \infty$. 
We introduced $\mup$ in our development because it allows for 
sharper statements and easier proofs.
\end{remark}

\begin{proposition}\label{prop:kappa-lip}
For~$F\in\Hd[q]$, the map 
$\IS^n \to\R,\, x\mapsto \kappa(F, x)^{-1}$ 
is $D$-Lipschitz continuous with respect to the Riemannian 
metric on $\IS^n$. 
\end{proposition}

\begin{proof}
Let~$x,y\in \IS^n$. Let~$u\in \scO(n+1)$ be the rotation that 
maps~$x$ to~$y$ and that is the identity on~$\left\{ x,y \right\}^\perp$.
By the invariance of Weyl's inner product under the action 
of $\scO(n+1)$,
\[ 
    d(F, \Sigma_{y}) = d(F\circ u, \Sigma_x). 
\]
Since the function~$g\mapsto d(g, \Sigma_x)$ is $1$-Lipschitz,
we obtain with Theorem~\ref{thm:kappa-dist} that 
\begin{align*}
  \|F\| \left| \frac{1}{\kappa(F,x)} - \frac{1}{\kappa(F, y)} \right| 
  &= | d(F,\Sigma_x) - d(F,\Sigma_y)| \\
  &= | d(F,\Sigma_x) - d(F \circ u,\Sigma_x)| 
  \leq \| F  - F\circ u\|. 
\end{align*}
We conclude the proof with the next lemma.
\end{proof}

\begin{lemma}\label{lem:lipshitz-unitary-action}
For any~$F \in \Hd[q]$ and any~$x,y\in \IS^n$,
\[ 
     \|F - F\circ u\| \leq D\|F\| d_{\IS}(x,y), 
\]
where~$u\in \scO(n+1)$ is the unique rotation that maps~$x$ 
to~$y$ and leaves invariant~$\left\{ x,y \right\}^\perp$. 
\end{lemma}

\begin{proof}
We first notice that
\[ 
    \|F - F\circ u\|^2 = \sum_i \|f_i - f_i \circ u\|^2 
\]
so it is enough to prove the claim when~$q=1$.
  
We prove a corresponding, more general 
statement over~$\C$ and, to this end, we consider the
space of complex homogeneous coefficients of degree~$d$ 
endowed with
Weyl's Hermitian inner product. The latter is invariant under the
action of the unitary group $\scU(n+1)$, therefore, without loss of
generality, we may assume that the matrix of~$u$ is the diagonal
matrix $\mop{diag}(e^{i\theta}, e^{-i\theta}, 1, \dotsc)$,
where~$\theta = d_\IS(x,y)$. We write $f =\sum_{|\bfa|= d}c_{\bfa}
X^\bfa$ and then 
\[ 
    f - f\circ u = \sum_{|\bfa|= d} \left( 1 - e^{i(a_0-a_1) \theta} \right) 
    c_{\bfa} X^\bfa. 
\] 
Since 
\[ 
   \left| 1 - e^{i(a_0-a_1) \theta} \right| \leq 
  \left| (a_0-a_1)\theta \right| \leq D \, \theta, 
\] 
we obtain the claim.
\end{proof}

\subsubsection{Condition number of homogeneous semialgebraic 
systems}\label{sec:kappa*}

We consider (closed) {\em homogeneous semialgebraic systems}, i.e., 
systems of the form 
\begin{equation}\label{eq:hsa}
  f_1(x)= 0,\ldots, f_q(x) = 0 \text{ and } g_1(x) \geq 0, \ldots, 
 g_s(x)\geq 0,  
\end{equation}
where the $f_i$ and the $g_j$ are homogeneous polynomials in
$\R[X_0,X_1,\ldots,X_n]$. The system is an element~$(F,G) \in
\Hd[q\,;\,s]$.
The set of solutions $x\in\IS^n$ of system~\eqref{eq:hsa}, 
which we will denote by $S(F,G)$, is 
a {\em spherical basic semialgebraic set}. Needless to say, we do
allow for the possibility of having $q=0$ or $s=0$. This corresponds 
with systems having only inequalities (resp. only equalities).

To a homogeneous semialgebraic system $(F,G)$ we associate 
a condition number~$\kappasa(F, G)$ as follows. For a 
subtuple~$L =(g_{j_1},\dotsc,g_{j_\ell})$ of~$G$, let~$F^L$ denote 
the system obtained from $F$ by appending the polynomials from 
$L$, that is,
\[ 
   F^L \eqdef \left(f_1,\dotsc,f_q,g_{j_1},\dotsc,g_{j_\ell}\right) 
   \in \Hd[q+\ell] 
\] 
(where now $\bfd$ denotes the appropriate degree 
pattern in $\N^{q+\ell}$). 
Abusing notation, 
we will frequently use set notations~$L\subseteq G$ or~$g\in G$ 
to denote subtuples or coefficients of~$G$.

\begin{definition}\label{def:kappaL}
Let $q\leq n+1$, $(F, G)\in\Hd[q\,;\,s]$.
The {\em condition number of the homogeneous semialgebraic 
system $(F,G)$} is defined as 
\begin{equation*}\label{eq:kappaL}
   \kappasa(F,G) \eqdef \max_{\substack{L\subseteq G\\ 
   q+|L| \leq n+1}} \kappa\big(F^L\big).
\end{equation*}
We define $\Sigma_*$ as the set of all~$(F,G)\in \Hd[q\,;\,s]$ such 
that~$\kappasa(F,G) = \infty$.
\end{definition}

Clearly, $\Sigma_*$ is semialgebraic and invariant under scaling 
of the $q+s$ components.

\begin{theorem}\label{thm:kappasa-dist}
For any nonzero~$\psi = (F, G)\in \Hd[q\,;\,s]$,
\begin{equation*}
    \kappasa(\psi)\leq \frac{\|\psi\|}{d(\psi, \Sigma_*)}.
\end{equation*}
\end{theorem}

\begin{proof}
For a subset~$L$ of the indices~$\left\{ 1,\dotsc,s \right\}$, let 
$p_L\colon\Hd[q\,;\,s]\to \Hd[q+|L|]$ be the projection~$(F,G) \mapsto F^L$. 
Clearly $\Sigma_* = \cup_L p_L^{-1}(\Sigma_L)$, 
where~$\Sigma_L$ is the set of ill-posed data  
in the appropriate space~$\Hd[q+|L|]$. 
In particular~$d(\psi, \Sigma_*) \leq d(p_L(\psi), \Sigma_L)$. 
Then, by Theorem~\ref{thm:kappa-dist}, 
\begin{equation*}
  \kappasa(F,G) = \max_{\substack{L\subseteq G\\ q+|L| \leq n+1}}  
  \frac{\|F^L\|}{d(F^L, \Sigma_L)} \leq 
  \frac{\|\psi\|}{d(\psi, \Sigma_*)}. \qedhere
\end{equation*}
\end{proof}

Note that we do not define condition for the very overdetermined 
case $q>n+1$, but it is important to include 
the overdetermined case $q+|L| = n+1$ in the definition 
of $\kappasa(F,G)$. To see why, consider the case of three 
polynomials $f,g_1,g_2$ around a point $x\in\IS^2$ as 
in Figure~\ref{fig:3curv}. 

\begin{figure}[t]\centering
\begin{tikzpicture}[scale=1.8]
  \begin{scope}
    \clip (0,0.4) rectangle (4,3.4);
    \begin{scope}
      \clip (4,-1.8) circle(5);
      \clip (4,4.5) circle(4);
      \fill[gray!20!white] (4,-1.8) circle(5);
    \end{scope}
    \draw[thick] (4,-1.8) circle(5);
    \draw[thick] (4,4.5) circle(4);
    \draw[thick] (4.79,1.5) circle(4);
  \end{scope}
 \path (1.5,3.3) node{\small $f=0$};
 \path (0.53,3.3) node{\small $g_1=0$};
 \path (-0.3,1.3) node{\small $g_2=0$};
 \path (0.67,2.08) node{\small $x$};
\end{tikzpicture}
\caption{The shaded area is where $g_1\geq 0$ and 
$g_2\geq 0$. Locally, the only solution point is the 
intersection $\{x\}$ of $f=g_1=g_2=0$.}
\label{fig:3curv}
\end{figure}
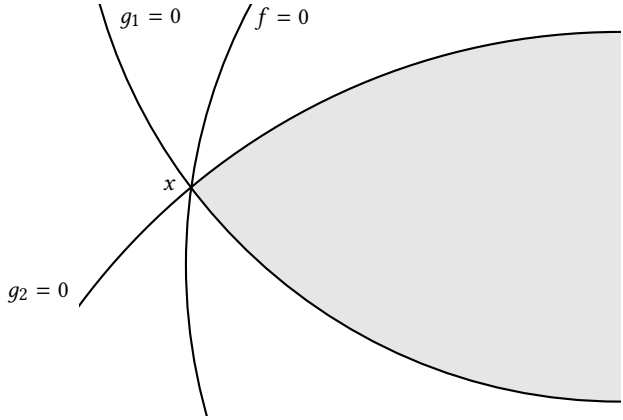
This system is ill-posed as arbitrarily small perturbations of 
$(f,G)$ may result in an empty intersection around $x$, 
and hence, a different topology of $S(f,G)$. But none of the 
condition numbers $\kappa(f,g_1)$ and $\kappa(f,g_2)$ capture 
this fact as $x$ is a well-posed zero for both systems.  
The following lemma is related to this matter.

\begin{lemma}\label{lem:Sfl-empty}
Let $(F, G)$ be a homogeneous semialgebraic system with 
$\kappasa(F, G) <\infty$. For any~$L\subseteq G$ such 
that~$|L| \geq n+1-q$, the set~$S(F^L,\varnothing)$ is empty. 
\end{lemma}

\begin{proof}
We choose $L' \subseteq L$ such that~$|L'| = n+1-q$. Because 
of the dimensions involved, $\Diff F^{L'}(x)|_{T_x}$ cannot be surjective,
thus~$\kappa(F^{L'},x) = \|F^{L'}\|/\|F^{L'}(x)\|$ for any~$x\in\IS^n$.
Moreover $\kappa(F^{L'}) \leq \kappasa(F, G) < \infty$, by definition
of~$\kappasa$. Therefore~$F^{L'}$ has no zero on~$\IS$.
In particular $S(F^L,\varnothing)$, the zero set of~$F^L$, is empty.
\end{proof}

We elaborate on Theorem~\ref{thm:tau-gamma-local}, 
relating~$\gamma$
and~$\tau$, to obtain the following result that relates~$\kappasa$
and~$\tau$.  It gives a computational handle on~$\tau$ which is
otherwise hard to get.

\begin{theorem}\label{thm:tau-kappa}
For any homogeneous semialgebraic system $(F,G)$ defining 
a semialgebraic set~$S \eqdef S(F, G) \subseteq \IS^n$, if 
$\kappasa(F,G)<\infty$, then
\begin{equation*}
    D^{\frac32}\, \tau(S)\, \kappasa(F, G) \geq \tfrac{1}{7}.
\end{equation*}
\end{theorem}

\begin{proof}
We first study the case where~$G = \varnothing$. 
Let~$\widehat S \subseteq\R^{n+1}$ be the cone over~$S$, 
that is, the zero set of~$F$ in~$\R^{n+1}$. For
any~$x\in S$, $\tau(S, x) \geq \min(1,\tau(\widehat S, x))$, by
Lemma~\ref{lem:tau-spherical}.
Therefore,
$$
 \tau(S) = \min_{x\in S} \tau(S,x) \ge \min\big(1, \min_{x\in S} 
 \tau(\widehat S, x) \big).
$$
Using also that $\kappa(F,x) \geq 1$
(Corollary~\ref{coro:kappa-geq-1}), we obtain that
\begin{equation}\label{eq:IM}
    \tau(S) \kappa(F) \geq \min\big(1, \min_{x\in S} \tau(\widehat S, x) 
    \kappa(F, x)\big).
\end{equation}
Recall from Lemma~\ref{le:mu-proj-norm} that
$\kappa(F, x) = \mup(F, x) = \mun(F, x)$
for all~$x\in S$. Combining this with Proposition~\ref{prop:gamma-mu},
we obtain that
$D^{\frac32} \kappa(F, x) \geq 2\gamma(F, x)$
for all $x\in S$. We conclude that 
$$
 D^{\frac32}  \min_{x\in S} \tau(\widehat S, x) \kappa(F, x)
 \ \ge\ \min_{x\in S} \, 2 \tau(\widehat S, x) \gamma(F, x) 
 \ \ge\ \tfrac17 ,
$$
where we have applied Theorem~\ref{thm:tau-gamma-local}
to $\widehat S$ for the right-hand side inequality.
Combining this with~\eqref{eq:IM}, we obtain 
$D^{\frac32} \tau(S) \kappa(f) \ge \tfrac17$.

We turn now to the general case $S \eqdef S(F, G) \subseteq \IS^n$. 
For $g\in G$ we define 
$P_g \eqdef \left\{x\in\IS^{n}\st g(x) \geq 0 \right\}$ and 
$W \eqdef S(F, \varnothing)$ so that
$S = W\cap \left( \cap_{g\in G} P_g \right)$.
We claim that for any~$L\subseteq G$,
\begin{equation}\label{eq:8}
    W \cap \bigcap_{g\in L} \partial P_g = S(F^L, \varnothing).
\end{equation}
The left-to-right inclusion is clear since~$\partial P_g$ is contained 
in the zero set of~$g$. Conversely, let~$x\in S(F^L, \varnothing)$ 
(in particular, $q+|L| \leq n$, by Lemma~\ref{lem:Sfl-empty}). 
The derivative $\Diff F^L(x)$ is surjective, because 
$\kappa(F^L, x) < \infty$. In particular, for any~$g \in L$, 
$Dg(x) \neq 0$ and since~$g(x) = 0$ it follows that the sign of~$g$
changes around~$x$. Thus~$x\in \partial P_g$ and 
Equation~\eqref{eq:8} follows.

Theorem~\ref{thm:tau-boundaries} implies that
\begin{equation}\label{eq:tau*}
    \tau(S) \ \geq\ \min_{L \subseteq G}
    \tau\Big(W \cap \bigcap_{g\in L} \partial P_g\Big)
    = \min_{L \subseteq G} \tau\big(S\big(F^L,\varnothing\big)\big) .
  \end{equation}
It suffices to take the minimum over the $L\subseteq G$ such 
that $q+|L| \le n+1$ because~$S(F^L,\varnothing) = \varnothing$ 
for larger~$L$. We obtain from the case~$G = \varnothing$ above,
\[
    7\, D^{\frac32} \tau(S) \ \geq\ \min_{L} 7\, D^{\frac32} \tau(S(F^L,
    \varnothing)) \ge \frac{1}{\max_{L} \kappa(F^L)} = 
    \frac{1}{\kappasa(F, G)},
\]
which completes the proof.
\end{proof}

\subsubsection{Strict inequalities}\label{se:relax}

We prove here that replacing inequalities 
$g_i(x) \ge 0$ by strict inequalities $g_i(x) > 0$
in the definition~\eqref{eq:hsa} of a spherical basic set~$S(F,G)$ 
does not change its homotopy type, provided $\kappasa(F,G) < \infty$. 

The argument is based on a general reasoning in topology. 
Recall that a closed subset $B$ of a topological space~$X$ is 
called {\em collared in $X$} if there exists a homeomorphism
$h\colon [0,1) \times B \to V$ 
onto an open neighborhood $V$ of $B$ in $X$ such that 
$h(0,b)=b$ for all $b\in B$. 

\begin{lemma}\label{le:collar}
If $B\subseteq X$ is collared in $X$ and 
$X\setminus B \subseteq X' \subseteq X$, 
then $X'$ and $X$ are homotopically equivalent. 
\end{lemma}

\begin{proof}
Let $\tau\colon V \to [0,1]$ and $u\colon V \to B$ 
denote the components of the inverse of $h$, so that 
$h(\tau(x), u(x)) = x$ for any~$x \in V$. 
We define the map~$\phi: [0,1]\times X\to X$ by
\begin{equation*}
    \phi_t(x) :=
    \begin{cases}
      h(t,u(x)) & \text{if $x\in V$ and~$\tau(x) < t$,}\\
      x & \text{otherwise.} 
    \end{cases}
\end{equation*}
The idea is that~$\phi_t$ pushes~$X$ increasingly far away 
from~$B$ as~$t$ increases.
It is easy to verify that $\phi$ is continuous, $\phi_0 = \Id_X$, 
$\phi_t(x)= x$ for $x\in X\setminus V$, and 
$\phi _1 (X) = X\setminus V$.
In other words, $\phi_t\colon X\to X$ defines a deformation 
retraction of $X$ to $X\setminus V$.

Moreover, we have $\phi_t(X') \subseteq X'$, since 
$\phi_t(X') \subseteq X\setminus B \subseteq X'$
for $t>0$. In addition, 
$X\setminus V \subseteq \phi_t(X') \subseteq X\setminus V$. 
Therefore, the restrictions of $\phi_t$ define a deformation 
retraction of $X'$ to $X\setminus V$.
We conclude that $X'$ and $X$ are homotopically equivalent.
\end{proof}

We apply this now to basic semialgebraic sets.

\begin{proposition}\label{pro:collar}
Let $(F,G) \in \Hd[q\,;\,s]$ be such that $\kappasa(F, G)<\infty$. 
Put $S:= S(F,G)$, let $r\le s$, and 
let $S'\subseteq S$ be 
the solution set in $\IS^n$ of the semialgebraic system
\[ 
    f_1=\dotsb=f_q=0,\ g_1 \geq 0,\dotsc,g_r \geq 0 \text{ and } 
   g_{r+1} >0,\dotsc,g_s > 0 .
\]
Moreover, let~$\partial S$ denote the boundary of~$S$ 
in~$S(F,\varnothing)$. Then 
$S\setminus \partial S \subseteq S'$, 
$\partial S$ is collared in~$S$, and 
$S'$ is homotopically equivalent to~$S$.
\end{proposition}

\begin{proof}
Let~$x\in S$ and~$L \subseteq G$ be maximal such that~$F^L(x) =0$. 
Note, this implies $|L|\le n-q$. 
Since~$\mup(F^L, x) = \kappa(F^L, x) < \infty$, the derivatives 
at~$x$ of the components of~$F^L$ are linearly independent. 
Therefore, the components of~$F^L$ are part of some regular 
system of parameters~$(f_1,\dotsc,f_q,v_1,\dotsc,v_{n-q})$
of~$\IS^n$ at~$x$ such that~$S$ is defined locally around $x$ by
\[ 
   f_1 = \dotsb = f_q = 0 \text{ and } 
   v_{1} \geq 0, \dotsc, v_{|L|} \geq 0, 
\]
and~$\partial S$ is defined locally around $x$ by 
additional requiring $v_j(x)=0$ for some $j\le |L|$. 
Therefore, if $x\not\in \partial S$, we must have $v_i(x)>0$ for all~$i$,  
and hence $g_j(x) > 0$ for all~$j$. 
This shows the first assertion $S\setminus \partial S \subseteq S'$. 

This reasoning also proves that locally around~$x$, 
the set $S$ is diffeomorphic to 
some~$(-1,1)^a \times [0,1)^b$ with~$a, b\in\mathbb{N}$.
Therefore, $\partial S$ is locally collared in~$S$.
By Brown's Collaring Theorem~\cite{Brown_1962,Connelly_1971},  
$\partial S$ is collared in~$S$, which proves the second assertion. 
The third assertion follows by applying Lemma~\ref{le:collar} 
to $X=S$, $B=\partial S$ and $X'=S'$.
\end{proof}

\subsubsection{Condition number of affine semialgebraic systems}
\label{sec:cond-numb-affine}

We now consider basic semialgebraic subsets of $\R^n$, rather
than~$\IS^n$. Given a degree pattern $\bfd=(d_1,\ldots,d_{q+s})$, the
homogeneization of polynomials (with respect to that pattern) yields
an isomorphism of linear spaces 
$$ 
   \Pd[q\,;\,s]\to\Hd[q\,;\,s],\quad
    \psi=(F,G)\mapsto \psi\hm = (F\hm,G\hm) , 
$$ 
where $F^h$ denotes the homogeneization
of $F$ with homogeneizing variable~$X_0$. The Weyl inner product 
on $\Hd[q\,;\,s]$ induces an inner product on $\Pd[q\,;\,s]$ 
such that the above map is isometric.

\begin{definition}
Let $(\bfd,1) := (d_1,\ldots,d_{q+s},1)\in\N^{q+s+1}$ be the 
degree pattern obtained from $\bfd$ by appending~1. Consider 
the \emph{scaled homogeneization map}
\begin{align}\label{eq:scH}
    H : \Pd[q\,;\,s] &\to\mcH_{(\bfd,1)}[q\,;\,s+1],\quad 
   \psi \mapsto (\psi\hm, \|\psi\hm\| X_0), 
\end{align}
that is, the system $H(F,G)$ is the homogeneization of~$(F,G)$ 
to which we add the inequality $X_0 \geq 0$ with a suitable 
coefficient. For $\psi\in\Pd[q\,;\,s]$, we define 
$\kappaaff(\psi)\eqdef \kappasa(H(\psi))$ and
call~$\Sigmaaff \eqdef H^{-1}(\Sigma_*)$ the 
\emph{set of ill-posed affine semialgebraic systems}. 
\end{definition}

By construction, $\|H(\psi)\|^2 = 2\|\psi\|^2$. 
We note that $\Sigmaaff$ is a semialgebraic set in~$\Pd[q\,;\,s]$ that is
invariant under scaling of each of the $q+s$ components. 

\begin{proposition}\label{prop:kappaaff-dist}
For any nonzero~$\psi\in \Pd[q\,;\,s]$,
\begin{equation*}
    \kappaaff(\psi) \leq \frac{4 D \|\psi\|}{d(\psi, \Sigmaaff)}.
\end{equation*}
\end{proposition}

\begin{proof}
Fix $\psi\in \Pd[q\,;\,s]$ and put $r:=\|\psi\|>0$. 
Further, assume~$\Phi \in \mcH_{\bfd,1}[q\,;\,s+1]$ is an element of~$\Sigma_*$ 
that minimizes the distance to $H(\psi)$. Theorem~\ref{thm:kappasa-dist} implies that
\begin{equation}\label{eq:11}
    \kappaaff(\psi) = \kappasa(H(\psi)) \leq \frac{\| H(\psi) \|}{\|H(\psi) - \Phi\|}.
\end{equation}
We write~$\Phi = (\Phi_1, \lambda)$, where~$\Phi_1 \in \Hd[q\,;\,s]$ 
and~$\lambda\in\mcH_1[1]$. 
We may assume that~$\lambda \neq 0$: otherwise, 
we can replace ~$\Phi$ with 
$\Phi' \eqdef (0, r X_0)$, which is an element of~$\Sigma_*$ 
that is at least as close to~$H(\psi)=(\psi\hm,rX_0)$ as~$\Phi$, since 
$\|H(\psi)- \Phi'\| = r \le \|H(\psi)- \Phi\|$. 

Since~$\Sigma_*$ is invariant under the scaling of each component, 
the minimality of~$\Phi$ implies that~$\lambda$ and~$rX_0 - \lambda$ are
orthogonal, that the angle~$\alpha$ between~$\lambda$ and~$X_0$ 
satisfies $\alpha \le \pi/2$, 
and, as a consequence, that
\begin{equation}\label{eq:10}
    \alpha \leq \tfrac\pi2 \sin\alpha \leq \tfrac\pi2 
    \tfrac{\|rX_0 - \lambda\|}{r} .
\end{equation}

Let $u\in \scO(n+1)$ be the rotation 
that leaves~$\left\{ X_0,\lambda \right\}^\perp$ invariant and 
such that~$\lambda \circ u =\|\lambda \| X_0$.    
Then 
$\Phi\circ u = (\Phi_1\circ u, \lambda\circ u) \in \Sigma_*$
since $\Sigma_*$ is invariant under the action of~$\scO(n+1)$. 
If we write 
$\Phi_1\circ u = \varphi\hm$ with $\varphi\in \Pd[q\,;\,s]$, then  
$H(\varphi) = (\varphi\hm,\|\varphi\hm\| X_0)$ lies in $\Sigma_*$,  
since $\Sigma_*$ is invariant under the scaling of its last component 
and $\lambda\ne 0$. 
We therefore obtain,
\begin{align*}
    d(\psi, \Sigmaaff) &\leq \|\psi - \varphi \| 
   = \|\psi\hm - \varphi\hm \| \leq \| \psi\hm - \psi\hm\circ u \| + \|\psi\hm\circ u - \Phi_1\circ u\| \\
    &= \| \psi\hm - \psi\hm\circ u \| + \|\psi\hm - \Phi_1\| .
\end{align*}
By Lemma~\ref{lem:lipshitz-unitary-action} and 
Inequality~\eqref{eq:10}, we obtain that
  \[ \| \psi\hm - \psi\hm\circ u \| \leq 
    \alpha D r \leq \tfrac\pi2 D \|rX_0 - \lambda\|. 
\]
Since~$\|rX_0 - \lambda\|$ and $\|\psi\hm - \Phi_1\|$ are both bounded
by~$\|H(\psi) - \Phi\|$, we get 
\[ 
   d(\psi, \Sigmaaff) \leq \left(\tfrac\pi2 D + 1 \right ) \| H(\psi) - \Phi \| 
  = \left( \tfrac\pi2 D+1 \right ) \sqrt{2}\, \|\psi\| \frac{\| H(\psi) - \Phi \|}{\|H(\psi)\|}. 
\] 
We conclude with Inequality~\eqref{eq:11}.
\end{proof}

\subsection{Neighbourhoods of spherical basic semialgebraic sets}
\label{sec:exclusion-inclusion}

The goal of this section is to compare two natural ways 
of defining neighborhoods of a spherical 
semialgebraic set~$S(F, G)$: by relaxing the arguments 
of the polynomials in~$F$ and~$G$ (the common, 
tube-like neighborhood), or by relaxing their values. 

For a subset $A\subseteq\IS^n$ we denote by 
\[ 
  \mcU_\IS(A, r) \eqdef \left\{ x\in\IS \st d_\IS(x, A) < r \right\}
\]
the open {\em $r$-neighborhood} of $A$ with respect to the geodesic 
distance~$d_\IS$ on the sphere~$\IS^n$.
Also, for a homogeneous system $(F, G) \in \Hd[q\,;\,s]$
and $r> 0$, we define the \emph{$r$-relaxation of~$S(F, G)$}:
\begin{align*}
  \Approx(F,G,r) &\eqdef \left\{ x \in \IS^n \st \forall f\in F\, \
                   |f(x)| < \|f\| r  \text{ and  } \forall g\in G\, g(x) > - \|g\| r \right \} .
\end{align*}
It is clear that~$S(F, G)\subseteq\Approx(F, G, r)$ for any~$r > 0$.
It is easy to see that $\Approx(F, G, r)$ converges to~$S$ with 
respect to the Hausdorff distance, when~$r\to 0$. 
The next two results quantify more precisely
this behaviour in terms of the condition number~$\kappasa(F,G)$.
Recall,~$D$ denotes the maximum degree of the 
components of~$F$ and~$G$. 

\begin{proposition}\label{prop:approx-easy}
For any~$r > 0$,
\begin{equation*}
    \mcU_\IS\left(S(F,G), D^{-\frac12} r \right) 
   \subseteq \Approx(F,G,r).
\end{equation*}
\end{proposition}

\begin{proof}
For any homogeneous polynomial~$h$ of degree~$d$ and 
any~$x,y\in\IS^n$,
\begin{equation*}
    \left| h(x) - h(y) \right| \leq \sqrt{d}\, \|h\|\, d_\IS(x,y).
\end{equation*}
(This is shown in~\cite[Lemma~19.22]{Condition}. The additional 
hypothesis $d_\IS(x,y) < 1/\sqrt{2}$ there can be easily removed by 
splitting the path from~$x$ to~$y$ in smaller segments.) 
Hence, for any~$x\in S$ and~$y\in \IS^n$ such
that~$d_\IS(x,y) < \frac{1}{\sqrt{D}}r$, any $f\in F$ and~$g\in G$, we have~$|f(y)| \leq r \|f\|$ 
and~$g(y)> g(x) - r\|g\| \geq -r\|g\|$.
\end{proof}

\begin{lemma}\label{lem:overdet}
Let $H\subseteq L\subseteq G$ be such that $|H|=n-q+1$, and 
$0<r<\frac{1}{\kappa(F^H)}$. Then 
$\Approx(F^L,G\setminus L,r)=\varnothing$.
\end{lemma}

\begin{proof}
Since~$\kappa(F^H) < \infty$ we 
have~$S(F^H,\varnothing) = \varnothing$, by 
Lemma~\ref{lem:Sfl-empty}. Assume there is 
a point~$x \in \Approx(F^L, G\setminus L, r)$. Then, as  
$H\subseteq L$ we have that $|h(x)|\leq r \|h\|$ for all $h\in F^H$ 
and it follows that 
\begin{equation*}
    \frac{1}{\kappa(F^H)} \leq  
    \frac{1}{\kappa(F^H,x)} =\frac{\|F^H(x)\|}{\|F^H\|} 
   \leq r. 
\end{equation*}
This is in contradiction with the hypothesis on~$r$ and hence
$\Approx(F^L, G\setminus L, r)$ is empty.
\end{proof}

\begin{theorem}\label{thm:approx}
Let $q\leq n+1$. For any positive number 
$r < \big( 13 D^{\frac32} \kappasa^2 \big)^{-1}$ we have 
\[ 
   \Approx(F, G, r) \ \subseteq\  \mcU_\IS(S(F,G), 3 \kappasa r). 
\]
 \end{theorem}

\begin{proof}
We will abbreviate~$S \eqdef S(F, G)$ 
and~$\kappasa \eqdef \kappasa(F,G)$.
The proof is by induction on the difference $\ell:=n-q+1$  
between the number of variables and the number of equations. 
Before dealing with the basis of the induction, we 
note that the assumption on $r$ implies that~$\kappasa<\infty$. 

If $\ell=0$, 
then $\kappa(F)=\kappasa<\infty$ and, because of 
our hypothesis, $r<\frac{1}{\kappasa}$. We deduce 
from Lemma~\ref{lem:overdet}, with $L=H=\varnothing$ 
that $\Approx(F, G, r) =\varnothing$. The desired 
inclusion is therefore trivially true. 

Now we assume $\ell>0$, i.e., $q\leq n$, and 
consider a point~$x\in\Approx(F, G, r)$. 
It is enough to show that 
\begin{equation}\label{eq:TOSHOW}
  d_{\IS}(x,S) < 3 \kappasa r. 
\end{equation}
To do so, we focus on the set of inequalities
$$
   L := \big\{ g\in G \mid |g(x)|<r\|g\| \big\} .
$$ 
By construction, we have $x\in\Approx(F^L, G\setminus L, r)$, and 
moreover~$g(x)\geq r\|g\| > 0$ for all~$g\in G\setminus L$.
We further note that $|L|\leq n-q$, otherwise there would 
exist $H\subseteq L$
with $|H|=n-q+1$ and, we would use again 
Lemma~\ref{lem:overdet} to deduce that
$\Approx(F, G, r) =\varnothing$, in contradiction with the fact that
$x\in\Approx(F,G,r)$. We next divide by cases.

\paragraph{Case~1: $L \neq \varnothing$.}
As $|F^L|\leq n+1$ we may apply the induction hypothesis 
to the larger set $F^L$ of equations 
and the smaller set $G\setminus L$ of inequalities; 
note that $\kappasa(F^L,G\setminus L)\leq \kappasa(F,G)$ 
so the hypothesis on $r$ is still true for $(F^L,G\setminus L)$. 
The induction hypothesis yields 
\[ 
   \Approx(F^L, G\setminus L, r) \subseteq 
   \mcU_\IS \left( S(F^L, G\setminus L), 
  3\kappasa(F^L,G\setminus L) r \right) 
   \subseteq 
   \mcU_\IS \left( S, 3\kappasa r \right) .
\]
Hence we obtain~\eqref{eq:TOSHOW} 
and are done in this case. 

\paragraph{Case~2: $L=\varnothing$.}
We put $u\eqdef \|F(x)\|/\|F\|$.
Then $u\le r$ since $x\in \Approx(F,G,r)$. 
Moreover, $\kappasa u \le \kappasa r < \frac{1}{13}$ 
by assumption. 
By definition, 
$$
 \kappa(F, x)^2  \ge \frac{1}{\mup(F, x)^{-2} + u^{2}}
 \ge \frac12 \min \big\{\mup(F, x)^{2}, u^{-2} \big\} . 
$$
This minimum equals $\mup(F,x)^2$ since 
$\kappasa u \le \frac{1}{13}$, so we get
\begin{equation*}\label{eq:excl-incl-3}
   \sqrt{2} \, \kappasa \ge \sqrt{2} \, \kappa(F) \ge \sqrt{2} \, \kappa(F,x) 
   \geq \mup(F, x) = \mun( \tilde{F}, x) ,
\end{equation*}
where $\tilde F \eqdef F|_{\mcT_x}$ denotes the restriction of~$F$ 
to the affine space~$\mcT_x$.
From the inequality above, 
Proposition~\ref{prop:gamma-mu}, and $u<r$, 
it follows that 
\begin{equation}\label{eq:alfa-beta}
\begin{split}
    \alpha(\tilde F,x) &\leq \frac12\, D^{\frac32} \mun(\tilde{F},x)^2 u 
        \le D^{\frac32}  \kappasa^2 r , \\ 
    \beta(\tilde F,x) &\leq \mun(\tilde{F},x) u \le \sqrt{2}  \kappasa r.
\end{split}
\end{equation}
From the assumption on~$r$,
we get~$\alpha(\tilde F, x)\leq \frac{1}{13}$, which makes possible 
the application of Theorem~\ref{thm:cont-alpha}. 
We also note that $\beta(\tilde F,x) < \frac{1}{13}$. 
As in~\S\ref{sec:cont-alpha-theory}, we define~$x_t$ in the affine 
space~$\mcT_x$ by the system of differential equations
\begin{equation*}
    \dot x_t = -\Diff\tilde F(x_t)^\dagger \tilde F(x_t), \quad x_0 = x.
\end{equation*}
Note that $x_t\ne 0$ for all $t\ge 0$ as $\|z\|\geq 1$ for 
all $z\in\mcT_x$. 
We define~$y_t \eqdef x_t / \|x_t\| \in \IS^n$. 
By Theorem~\ref{thm:cont-alpha}, there is a limit 
point~$x_\infty \in \mcT_x$, which is a zero of~$\tilde F$, 
and which satisfies $\|x_\infty - x\| < 2\beta(\tilde F, x)$. 
In particular, $y_\infty$ is a zero of~$F$ and 
\[ 
   d_\IS(y_\infty, x) \leq \|x_\infty-x\| \le 2 \beta(\tilde{F},x) \le 
  2 \sqrt{2} \kappasa r < 3 \kappasa r ,
\]
where we used~\eqref{eq:alfa-beta} for the second inequality. 
If~$g(y_\infty) \geq 0$ for all~$g\in G$, then~$y_\infty\in S$ and 
$d_{\IS}(x, S) \leq d_{\IS}(x, y_\infty)$, hence~\eqref{eq:TOSHOW} and we are done. 

So suppose that~$g(y_\infty) < 0$ for some~$g\in G$ 
and let~$s > 0$ be the smallest real  number such that~$g(y_s) = 0$ 
for some~$g\in G$. By construction, the set 
$H \eqdef \left\{ g\in G \st g(y_s) = 0 \right\}$ is nonempty 
and element of~$G\setminus H$ is positive 
at~$y_s$. Also, for every~$f\in F$, 
$$
 |f(y_s)| = \frac{|f(x_s)|}{\|x_s\|^{\deg f}} \le  
  |f(x_s)| = |f(x)| e^{-s}  \le \|f\| r e^{-s} 
$$
where the second equality is due to 
Theorem~\ref{thm:cont-alpha}\ref{item:2}.
Therefore, $y_s \in \Approx(F^H, G\setminus H, r e^{-s})$.

Using again Lemma~\ref{lem:overdet} we deduce 
that $|H|<n-q+1=\ell$. We can therefore apply the 
induction hypothesis to the larger set $F^H$ of equations 
and the smaller set $G\setminus H$ of inequalities; 
note that $re^{-s} < r$ and 
$\kappasa(F^H,G\setminus H) \leq \kappasa$. 
Thus we obtain
$$
  \Approx(F^H, G\setminus H, r e^{-s}) \ \subseteq\  
  \mcU_\IS \left( S(F^H,G\setminus H), 
  3\kappasa(F^H,G\setminus H) r \right)
  \ \subseteq\ \mcU_\IS \left(S, 3\kappasa r \right), 
$$
the latter because $S(F^H,G\setminus H)\subseteq S$ and 
$\kappasa(F^H,G\setminus H) \leq \kappasa$. 
We conclude that 
$$
  d_{\IS}(y_s,S) < 3 \kappasa r e^{-s}.
$$
Also, by Theorem~\ref{thm:cont-alpha}\ref{item:5}, 
$$
   d_{\IS}(y_s,x)\leq \|x_s-x\|\leq 2\beta(\tilde{F},x)(1-e^{-s})
   < 2\sqrt{2} \kappasa r (1-e^{-s}),
$$ 
the last inequality by~\eqref{eq:alfa-beta}. 
We finally deduce that 
\[ 
    d_{\IS}(x, S) \leq d_\IS(x, y_s) + d_\IS(y_s, S) < 
    \big(2\sqrt{2} (1-e^{-s}) + 3 e^{-s} \big) \kappasa r
    <  3 \kappasa r, 
\]
which shows~\eqref{eq:TOSHOW} and finishes the proof. 
\end{proof}

\subsection{The geometry of ill-posedness} 
\label{sec:ill-posed}

In order to analyze the set $\Sigma\subseteq\Hd[q]$ of 
ill-posed inputs, cf.~\eqref{def:Sigma}, 
we first study 
its complex version, defined as 
$$
 \SigmaC := \big\{ F\in \HdC[q] \mid \exists x\in\P^n 
  \; F(x) = 0,\; \mathrm{rank\,}\Diff F(x)_{\mid T_x} < q \big\} .
$$
Here $\P^n$ denotes the complex projective space of dimension~$n$. 
Note that because of Euler's formula~\cite[(16.3)]{Condition} we have 
$$
 \SigmaC := \big\{ F\in \HdC[q] \mid \exists x\in\P^n 
  \; F(x) = 0,\; \mathrm{rank\,}\Diff F(x)< q \big\} .
$$
In the special case $q=n+1$, we have  
$\SigmaC := \{ F\in \HdC[n+1] \mid \exists x\in\P^n \; F(x) = 0\}$.
It is well known that this is 
the zero set of the multivariate resultant, 
which is an irreducible polynomial with integer coefficients 
and degree 
$\sum_{i=1}^{n+1} \prod_{k\ne i} d_k$ 
\cite[\S13.1]{GelfandKapranovZelevinsky_1994}, 
which is at most~$(n+1)D^n$. An extension 
of this result to the case $q\leq n$ appears 
in~\cite[Proposition~5.3]{CuckerKrickShub_2018}.   
We further generalize this result to $q\le n+1$, slightly 
improving the bound in passing.

\begin{proposition}\label{prop:discriminant}
For any~$q \leq n+1$, the variety~$\SigmaC \subseteq \HdC[q]$ is 
a hypersurface defined by an irreducible polynomial with integer 
coefficients of degree at most $n2^n D^n$. 
\end{proposition}

\begin{proof}
We abbreviate $\Hc:=\HdC[q]$ and 
first assume $q\le n$. 
Consider the incidence variety 
\begin{equation}\label{defSigmaCtilde}
   \tilde\SigmaC \eqdef \left\{ (F,x,v) \in 
   \Hc\times (\C^{n+1}\setminus \{0\}) \times (\C^{q}\setminus \{0\}) 
   \st F(x) =0 
   \text{ and } v^T \cdot \Diff F(x) = 0 \right\}. 
\end{equation}
The projection 
$\tilde\SigmaC \to (\C^{n+1} \setminus \left\{ 0 \right\}) 
 \times(\C^q \setminus \left\{ 0 \right\}),\, 
 (F,x,v)\mapsto (x,v)$ is surjective.
Moreover, the fibers are
linear subspaces of~$\Hc$ of codimension~$n+q$ 
This implies that $\tilde\SigmaC$ is irreducible~\cite[\S6.3,
Thm.~8]{ShafarevichI} 
and $\dim \tilde\SigmaC = (n+1) +q + \dim\Hc -n-q = \dim\Hc +1$. 
The image of
the projection $\tilde{\SigmaC} \to \Hc,\, (F,x,v) \, \mapsto F$ equals
$\SigmaC$, which is therefore irreducible. Moreover, since the fibers 
of this
projection are generically of dimension~2, 
it follows that $\dim\SigmaC = \dim \tilde{\SigmaC} - 2 = \dim\Hc -1$. 
Hence
$\SigmaC$ is indeed an irreducible hypersurface in $\Hc$. 
That its defining
equation has integer coefficients follows from elimination
theory~\cite[\S2.C]{Mumford} and the fact that $\tilde{\SigmaC}$ 
is defined by polynomials with integer coefficients.
  
For bounding the degree, we consider a variant of $\tilde{\SigmaC}$ 
in a product of projective spaces. 
More specifically, we consider the variety~$S$ 
of all~$(F,x,u)\in \P(\Hc)\times \P^n\times \P^{q-1}$,  
which are solutions of the multihomogeneous equations
\begin{equation}\label{eq:6}
   \left\{
   \begin{aligned}
      & f_i(x) = 0 && \text{for } 1\leq i\leq q ,\\
      & \sum_{i=1}^q u_i \, x_0^{D-d_i} \frac{\partial f_i}{\partial x_j}(x) 
          && \text{for } 1\leq j \leq n .
      \end{aligned}
    \right.
\end{equation}
We note that the projection $(F,x,u)\mapsto F$ maps 
$S\cap \{X_0\ne 0\}$ to $\SigmaC$ and hits all 
$F\in\SigmaC$ except those in a lower dimensional subvariety. 

We take now hyperplanes $H_1,\ldots,H_{N-1} \subseteq\Hc$ 
in general position, 
where $N:=\dim\P(\Hc)$. Let us denote by 
$\tilde{H}_k$ the inverse image of $H_k$ under the projection $(F,x,u)\mapsto F$. 
Then we have 
\begin{equation}\label{eq:deg-ineq}
 \deg\SigmaC = |H_1\cap\ldots\cap H_{N-1} \cap \SigmaC | \ \le\ 
   | \tilde{H}_1 \cap \ldots \cap \tilde{H}_{N-1} \cap S | =: M .
\end{equation}
The number~$M$ of intersections points on the right-hand side can 
be computed with the multiprojective B\'ezout's 
theorem, see e.g.~\cite[\S4.2.1]{ShafarevichI}\cite{MorganSommese_1987}. 
According to this, $M$ equals the coefficient of the monomial 
$a^N b^n c^{q-1}$ in the product ($a,b,c$ are formal variables) 
\begin{equation} \label{eq:ABC}
    a^{N-1}\prod_{i=1}^q   (a + d_i b) \prod_{i=1}^n ( a + (D-1) b + c). 
\end{equation}
For this, note that the equations for $\tilde{H}_k$ have the 
multidegree $(1,0,0)$,  
and the equations in~\eqref{eq:6} have the multidegree $(1,d_i,0)$ 
and $(1,D-1,1)$ with respect to (the coefficients of) $F$, $x$, 
and $u$, respectively. 
The coefficient $M$ can be bounded as 
$$
  M \ \le\ 
  q {n\choose q-1} D^{q-1} D^{n-q+1}  + n {n-1\choose q-1} D^{q} 
  D^{n-q}\ \le\ n 2^n D^n .
$$
Indeed, when expanding~\eqref{eq:ABC}, the left-hand contribution 
arises
from selecting $a$ in exactly one of the $q$ factors in the left product 
and selecting $c$ in  exactly $q-1$ among the $n$ factors of the 
right product.
The right-hand contribution arises from selecting $a$ in exactly one
of the $n$ factors in the right product 
and selecting $c$ in  exactly $q-1$ among the remaining $n-1$ factors of the right product.

In the case $q=n+1$ we consider the incidence variety 
$S:=\{ (F,x) \mid F(x)=0\} \subseteq \P(\Hc)\times \P^n$
and argue similarly. In particular, the multiprojective B\'ezout's theorem 
implies that $\deg\SigmaC$ equals the coefficient of the monomial
$a b^n$ in the product 
$\prod_{i=1}^{n+1}(a + d_i b)$. 
This leads to the well known formula 
$\deg\SigmaC =\sum_i \prod_{k\ne i} d_i$. 
Since this is bounded by $(n+1)D^{n}$, 
the degree bound in this case follows as well. 
\end{proof}

The weaker bound $\deg\SigmaC \le D^{q+n}$, 
which is good enough for our purpose, can be obtained with a 
significantly simpler argument. 
From \eqref{defSigmaCtilde} we obtain with B\'ezout's Inequality
$\deg\tilde{\SigmaC} \le D^q \cdot D^n$ \cite[\S8.2]{BCSh}. 
(Note that on an open subset we only need $n$ equations out 
of $v^T \cdot \Diff F(x) = 0$.) We conclude that 
$\deg\SigmaC \le \deg\tilde{\SigmaC} \le 
D^{q+n}$~\cite[Lemma~8.32]{BCSh}.

\begin{corollary}\label{cor:degree-sigmasa}
The set $\Sigma_* \subseteq \Hd[q\,;\,s]$ of ill-posed homogeneous systems 
is included in the zero set of a nonzero polynomial with integer coefficients 
of degree at most $n2^n (s+1)^{n+1} D^n$. 
The same holds true for the set $\Sigmaaff \subseteq \Pd[q\,;\,s]$ 
of ill-posed affine systems.
\end{corollary}

\begin{proof}
For a subset $L = \left\{ i_1,\dotsc,i_{\ell} \right\}$ of 
$\left\{ 1,\dotsc,s \right\}$, let $p_L$ be the projection
\begin{equation*}
    p_L\colon \Hd[q\,;\,s] \to \Hd[q+\ell],\, 
   (f_1,\dotsc,f_q,g_1,\dotsc,g_s)\in \Hd[q\,;\,s] \mapsto 
    (f_1,\dotsc,f_q,g_{i_1},\dotsc,g_{i_{\ell}}) 
\end{equation*}
By definition of~$\Sigma_*$ and~$\kappasa$, $\Sigma_*$ 
is the union of
the sets~$p_L^{-1}(\Sigma_L)$ for all~$L$ with~$q+|L| \leq n+1$,
where~$\Sigma_L$ is the appropriate set of ill-posed data 
in~$\Hd[q+\ell]$.  The number of such subsets $L$ is at most
$(s+1)^{n+1-q}$ and we conclude with 
the fact that, for each of them, 
$\Sigma_L\subseteq\Sigma_L^{\C}\cap\Hd[q+\ell]$ 
and the latter is the set of real zeros of a polynomial 
of degree $n2^nD^n$ by 
Proposition~\ref{prop:discriminant}.

To settle the affine case, note that the scaled homogeneization map~\eqref{eq:scH} 
has the structure 
$H\colon\R^N\to\R^N\times\R,\, a\mapsto (a,\|a\|)$
and that $\Sigma_*$ is scale invariant. By definition, $\Sigmaaff = H^{-1}(\Sigma_*)$. 
Suppose that the polynomial~$P$ vanishes on $\Sigma_*$
and let $P(a,t) =\sum_i P_i(a) t^i$ 
be its decomposition into homogeneous parts.  
Then each $P_i$ vanishes on $H^{-1}(\Sigma_*)$.
\end{proof}

\section{Algorithms}
\label{sec:an-algor-comp}

\subsection{The covering algorithm}
\label{sec:algorithm}

The main stepping stone towards computing the homology groups 
of a spherical semialgebraic set $S$ is the computation of a finite set 
$\mcX$ and a real $\e>0$ such that $S$ is
homotopically equivalent to $U_\e(\mcX)$. 
We will do so using Theorems~\ref{thm:SNW} and~\ref{thm:approx} 
in conjunction.

For~$0 < r < 1$ we define~$\mcG_r$ as the image in $\IS^n$
under the map $y\mapsto\frac{y}{\|y\|}$ of the set of  
points~$x\in\Z^{n+1}$
with~$\|x\|_\infty = \lceil \frac{\sqrt{n}}{r} \rceil$.
We easily check that 
\begin{equation}\label{eq:cover}
  \IS^n \subseteq \bigcup_{x\in \mcG_r} B_\IS(x, r) ,
\end{equation}
where $B_\IS(x, r) := \{y\in \IS^n \mid d_{\IS}(x,y) < r \}$. 
Moreover~$|\mcG_r| = (n/r)^{\mcO(n)}$.

\begin{proposition}\label{thm:costH}
On input~$F$ and~$G$, Algorithm \textsc{Covering} outputs a finite 
set~$\mcX$ and an $\e > 0$ such that~$\mcU(\mcX, \e)$ is 
homotopically equivalent to~$S(F, G)$.
Moreover, the computation performs 
$\big((s+n) D \kappasa\big)^{\mcO(n)}$ arithmetic operations, 
where $s=|G|$ and 
$\kappasa = \kappasa(F,G)$, 
and the number $|\mcX|$ of points in~$\mcX$ 
is~$(n D \kappasa)^{\mcO(n)}$.
\end{proposition}

\begin{algo}[t]
  \centering
  \begin{description}
    \item[Input.] A homogeneous semialgebraic 
      system~$(F,G) \in \Hd[q\,;\,s]$ with $q\le n$.
    \item[Precondition.]
        $\kappasa(F, G)$ is finite.
      \item[Output.]
        A finite subset~$\mcX$ of~$\IS^n$ and an~$\epsilon > 0$.
      \item[Postcondition.]
        $\mcU(\mcX, \epsilon)$ is homotopically
        equivalent to~$S(F, G)$.
  \end{description}
  \medskip
  \begin{algorithmic}
    \Function{Covering}{$F$, $G$}
    \State $r\gets 1$
    \Repeat
    \State $r \gets r/2$.
    \State $k_* \gets \max \left\{  \kappa(F^L,x) \st x \in \mcG_r 
    \text{ and } L \subseteq G \text{ such that } |L| \le n+1 -q \right\}$
    \Until{$71\, D^{\frac 52} k_*^2 r < 1$}    
    \State\Return the set 
    $\mcX:=\mcG_r \cap \Approx(F, G, D^{\frac12}r)$ and the real
    number $\e:= 5D k_* r$
    \EndFunction
\end{algorithmic} 
\caption[]{\scshape Covering} 
\label{algo:covering}
\end{algo}

\begin{proof} 
Let~$\kappasa \eqdef \kappasa(F, G)$, $S\eqdef S(F,G)$ and 
let~$r$ and~$k_*$ be the values of the corresponding variables 
after the \emph{repeat} loop terminates in Algorithm \textsc{Covering}. 
By design,
\begin{equation}\label{eq:4}
    71\, D^{\frac 52} k_*^2 r < 1.
\end{equation}
  
We will first show that
\begin{equation}\label{eq:3}
    \kappasa \leq (1+\tfrac1{100}) k_*.
\end{equation}
Let $L\subseteq G$ and $y \in \IS^n$ be such that 
$\kappasa = \kappa(F^L) =
\kappa(F^L, y)$. Because of~\eqref{eq:cover} there is 
some~$x\in\mcG_{r}$ such
that~$d_\IS(x, y) < r$, and $\kappa(F^L, x) \le k_*$ by the 
definition of $k_*$.
Since the map $x\mapsto 1/\kappa(F^L,x)$ is $D$-Lipschitz 
continuous (Proposition~\ref{prop:kappa-lip}), we have
\[ 
    \kappasa = \kappa(F^L, y) \leq 
    \frac{\kappa(F^L,x)}{1-D \kappa(F^L, x) r} \leq
    \frac{k_*}{1-Dk_*r}. 
\] 
Inequality~\eqref{eq:4} shows that 
$$
   Dk_*r < \frac{1}{71\, D^{\frac32} k_*} \leq \frac{1}{101} 
$$
the last as $D\geq 2$ and $k_*\ge1$, and Inequality~\eqref{eq:3} 
follows.
  
Let~$\mcX \eqdef \mcG_r \cap \Approx(F, G, D^{\frac12} r)$ 
and~$\epsilon\eqdef 5D k_* r$, that is, the finite set 
and the real number output by the algorithm. We will now prove 
that~$\mcU(\mcX, \epsilon)$ is homotopically equivalent 
to~$S$. By Theorem~\ref{thm:SNW}, 
it is enough to prove the inequalities
\begin{equation}\label{eq:5}
    3d_H(\mcX, S) < \epsilon < \frac{1}{2}\tau(S).
\end{equation}
The second inequality follows from 
Inequalities~\eqref{eq:4}, \eqref{eq:3} and 
Theorem~\ref{thm:tau-kappa}:
\[ 
   \e= 5D k_* r < \frac{5}{71}\,\frac{1}{D^{\frac32} k_*} \leq 
    \frac{505}{7100}\,\frac{1}{D^{\frac32} \kappasa} \ \le\ 
       \frac{3535}{7100} \tau(S) \le \frac{1}{2} \tau(S) .
\]
Concerning the inequality $ 3d_H(\mcX, S) < \epsilon$, let~$x\in S$. Because of~\eqref{eq:cover},
there is some~$y\in \mcG_{r}$ with $d_\IS(x,y) < r$. 
Hence~$y$ lies in $\Approx(F, G, D^{\frac12} r)$, 
by Proposition~\ref{prop:approx-easy}. 
Thus $y\in \mcX$ 
and~$d(x, \mcX)<d_\IS(x,y)< r<\frac13\e$.

Next, let~$x \in \mcX$. Then, $x\in\Approx(F, G, D^{\frac12} r)$ 
and 
\[ 
	13\, D^{\frac32} \kappa_*^2 (D^{\frac12} r) <
    71\,D^{\frac52} \kappa_*^2 r 
    <1 
\]
the last by Inequality~\eqref{eq:4}. Hence, 
Theorem~\ref{thm:approx} applies and shows that
\[ 
  d(x, S) \leq d_{\IS}(x, S) \leq 3 \kappa_* D^{\frac12} r 
  \leq (3 + \tfrac3{100}) k_* D^{\frac12} r < \frac{1}{3} \epsilon ,
\]
where we used $D\ge 2$ for the last inequality. 
Thus we have shown that 
$d_H(\mcX, S) < \frac13\epsilon$.
This concludes the proof of \eqref{eq:5} and of the homotopy equivalence.
\smallskip
  
Lastly, we deal with the complexity analysis. We can 
approximate~$\kappa(F^L, x)$ within a factor of~$2$
in~$\mcO(N+n^3)$ operations \cite[\S2.5]{Lairez_2017} and this 
is enough for
our needs. For simplicity, we will do as if we could compute 
$\kappa$ exactly.

The \emph{repeat} loop performs~$\mcO(\log(D\kappasa))$ iterations. 
Each iteration can be done in~$\mcO(|\mcG_r| M (N+n^3))$ operations,
where 
$M = \sum_{i=0}^{n+1-q} {s \choose i} \leq (s+1)^{n+1-q}$.
Moreover, 
$|\mcX|\le|\mcG_r|=(nD\kappasa)^{\mcO(n)}$ and 
$N+n^3=(nD)^{\Oh(n)}$.
Therefore, the total number of operations is bounded by 
$\big((s+n)D\kappasa\big)^{\Oh(n)}$.
\end{proof}

\subsection{Homology of a union of balls}
\label{sec:homology-union-balls}

Once in the possession of a pair $(\mcX,\e)$ such that 
$S$ is a deformation retract of $\mcU(\mcX,\e)$, the computation 
of the homology groups of $S$ is a known process. 
One computes the nerve~$\mcN$ of the covering 
$\{B(x,\e)\mid x\in\mcX\}$ (this is a
simplicial complex whose elements are the subsets~$N$ 
of~$\mcX$ such
that~$\cap_{x\in N} B(x,\e)$ is not empty) and from it, its 
homology groups 
$H_k(\mcN)$. Since the intersections of any collection 
of balls is convex, the Nerve Theorem \cite[e.g.][Thm.~10.7]{Bjo:95} 
ensures that 
$$
    H_k(\mcN)\simeq H_k(\mcU(\mcX,\e))\simeq H_k(S)
$$
the last because $S$ is a deformation retract of $\mcU(\mcX,\e)$.

The process is described in detail in 
of~\cite[\S4]{CuckerKrickShub_2018} where the proof for the following 
result can be found (see 
also~\cite{Edelsbrunner_1995,EdelsbrunnerShah_1992} for 
improved algorithms for computing the nerve of a covering).

\begin{proposition}\label{prop:comp-homology}
Given a finite set~$\mcX \subseteq \R^{n+1}$ and a positive real
number~$\epsilon$, one can compute the homology
of~$\cup_{x\in\mcX} B(x, \epsilon)$ with~$|\mcX|^{\mcO(n)}$ 
operations. \eproof
\end{proposition}

\subsection{Homology of affine semialgebraic sets}
\label{sec:affine}

A pair $(F,G)\in \Pd[q\,;\,s]$ defines a basic semialgebraic set
$W(F,G)\subseteq\R^n$ as in~\eqref{eq:sa}
which is diffeomorphic to the subset of $\IS^n$ 
defined by~$F\hm = 0$, $G\hm \succ 0$ and~$X_0 > 0$. As
in~\S\ref{sec:cond-numb-affine}, 
let~$H(F,G)\in\mcH_{(\bfd,1)}[q\,;\,s+1]$ denote this 
system of homogeneous polynomials 
(with $X_0 > 0$ replaced by $\|(F,G)\| X_0 > 0$, 
which does not change the solution set). 
Proposition~\ref{pro:collar} tells us that, 
unless this system is ill-posed,
we may replace~$X_0 > 0$ with~$X_0 \geq 0$ and 
any $g>0$ with $g\geq 0$ without changing the homology 
of the solution set. In other words, if
$\kappaaff(F,G) < \infty$, then the spherical set~$S(H(F,G))$ is
homotopically equivalent to~$W(F,G)$. 

Based on the tools introduced above, we may compute the homology 
of~$W(F,G)$, assuming that~$\kappaaff(F, G) < \infty$, by 
computing the nerve of a suitable
covering of~$S(H(F,G))$ obtained with 
Algorithm {\scshape covering}.  
This leads to Algorithm {\scshape homology} below whose
analysis will prove Theorem~\ref{thm:Main}.

\begin{algo}[t]
  \centering
  \begin{description}
    \item[Input.] A semialgebraic system~$(F, G) \in \Pd[q + s]$ with
       $q\le n$.
      \item[Output.]
       The homology groups of the set $\left\{ f_1 = \dotsb = f_q = 0 
       \text{ and } g_1 \succ 0, \dotsc, g_q \succ 0 \right\} 
       \subseteq \R^n$.
  \end{description}
\medskip
  \begin{algorithmic}
    \Function{Homology}{$F$, $G$}
    \State  $(\mcX,\e)\gets\text{\scshape Covering}(H(F,G))$
    \State $\mcN \gets$ the nerve of $\mcU(\mcX,\e)$
    \State \Return the homology groups of $\mcN$
    \EndFunction
\end{algorithmic}
\caption[]{{\scshape Homology}}
\label{algo:homology}
\end{algo}

\begin{proof}[Proof of Theorem \ref{thm:Main} \ref{item:9}]
By Proposition~\ref{thm:costH}, the cost of computing the 
covering~$\mcX$ is bounded 
by~$\big((s+n)D \kappaaff)^{\Oh(n)}$, where~$\kappaaff \eqdef
\kappaaff(F, G)$, and $|\mcX| = (nD\kappaaff)^{\Oh(n)}$. By
Proposition~\ref{sec:homology-union-balls}, the cost of computing 
the nerve $\mcN$ and its homology groups is $|\mcX|^{\Oh(n)}$. 
Hence, the total cost of the algorithm is bounded by 
$\big((s+n)D \kappaaff)^{\Oh(n^2)}$.
Together with Proposition~\ref{prop:kappaaff-dist}, this leads to the
conclusion.
\end{proof}

The probabilistic analysis is based on the following result by
Bürgisser and Cucker~\cite[Theorem~21.1]{Condition} 
and follows a line of similar results that rely on the 
same ideas. We will be consequently brief. We 
rephrased the statement in terms of the isotropic Gaussian 
distribution instead of the uniform distribution on the sphere. 
The scale invariance of the statement makes both
formulations equivalent.

\begin{theorem}\label{th:BCL}
Let~$\Sigma\subseteq \R^{p+1}$ be contained in a real algebraic 
hypersurface, 
given as the zero set of a homogeneous polynomial of degree~$d$
and let~$a \in \R^{p+1}$ be a centered isotropic Gaussian random 
variable. Then for all $t \geq (2d+1)p$,
\begin{equation}\tag*{\qed}
   \Prob\left( \frac{\|a\|}{d(a,\Sigma)} \geq t \right) \leq \frac{11 d p}{t}.
\end{equation}
\end{theorem}

\begin{proof}[Proof of Theorem~\ref{thm:Main} \ref{item:7} and 
\ref{item:8}]
Let~$\psi = (F,G)\in \Pd[q\,;\,s]$ be a centered isotropic Gaussian 
random variable. By Theorem \ref{thm:Main} \ref{item:9}, the number 
of operations performed by algorithm \textsc{Homology} is
\[ 
   \cost(\psi) = \left( (s+n) D \frac{\|\psi\|}{d(\psi, \Sigmaaff)}
    \right)^{C n^2}, 
\]
for some $C > 0$. 

By Theorem~\ref{th:BCL} and 
Corollary~\ref{cor:degree-sigmasa},
\begin{equation*}
  \Prob \left(\cost(\psi) \geq \left(\big((s+n) D t\big)^{C n^2}\right)\right) \leq 
  \frac{11 n 2^n (s+1)^{n+1} D^n N}{t} = \frac{ \left( (s+n)D \right)^{\Oh(n)}}{t},
\end{equation*}
where~$N \eqdef \dim\Pd[q\,;\,s] \leq (s+n) (D+1)^n$. We obtain
Theorem~\ref{thm:Main}\ref{item:7} with~$t = \left( (s+n)D \right)^{c n}$
and Theorem~\ref{thm:Main}\ref{item:8} with 
$t = 2^{cN}$, for some~$c$ large enough. 
For the latter, we use that $\big((s+n) D\big)^{\Oh(n)} = 2^{\Oh(N)}$ 
and that $n^2 = \Oh(N)$. 
\end{proof}

\begin{acks}
We are grateful to Josu\'e Tonelli-Cueto 
and Mohab Safey El~Din for helpful discussions, to Teresa Krick
for her careful reading, and to Théo Lacombe and the referees for relevant remarks.

This work has been supported by the Einstein Foundation, Berlin,
 by the \grantsponsor{BU137122}{DFG}{} research grant
 no.~\grantnum{BU137122}{BU~1371/2-2},
 and by the \grantsponsor{HK}{Research Grants Council of the Hong Kong
SAR}, project no.~\grantnum{HK}{CityU-11202017}.
\end{acks}

\bibliographystyle{ACM-Reference-Format}
\bibliography{maths,book}

\end{document}